\documentclass{lmcs} 
\pdfoutput=1
\usepackage[utf8]{inputenc}

\usepackage{lastpage}
\lmcsdoi{21}{3}{11}
\lmcsheading{}{\pageref{LastPage}}{}{}%
{Mar.~07,~2024}{Jul.~29,~2025}{}


\usepackage{hyperref}
\usepackage{quiver}

\usepackage{tikzit}

\tikzstyle{observation}=[text=blue]
\tikzstyle{morphism}=[fill=white, draw=black, shape=rectangle]
\tikzstyle{medium box}=[fill=white, draw=black, shape=rectangle, minimum width=0.8cm, minimum height=0.9cm]
\tikzstyle{large morphism}=[fill=white, draw=black, shape=rectangle, minimum width=1.7cm, minimum height=1cm]
\tikzstyle{bn}=[fill=black, draw=black, shape=circle, inner sep=1.5pt]
\tikzstyle{state}=[fill=white, draw=black, regular polygon, regular polygon sides=3, minimum width=0.8cm, shape border rotate=90, inner sep=0pt]
\tikzstyle{nn}=[fill=gray, draw=gray, regular polygon, regular polygon sides=3, minimum width=0.4cm, shape border rotate=90, inner sep=0pt]
\tikzstyle{effect}=[fill=white, draw=black, regular polygon, regular polygon sides=3, minimum width=0.8cm, shape border rotate=270, inner sep=0pt]
\tikzstyle{medium state}=[fill=white, draw=black, regular polygon, regular polygon sides=3, minimum width=1.3cm, inner sep=0pt, shape border rotate=90]
\tikzstyle{medium effect}=[fill=white, draw=black, regular polygon, regular polygon sides=3, minimum width=1.3cm, inner sep=0pt, shape border rotate=270]
\tikzstyle{large state}=[fill=white, draw=black, regular polygon, regular polygon sides=3, minimum width=2.2cm, shape border rotate=180, inner sep=0pt]
\tikzstyle{wn}=[fill=white, draw=black, shape=circle, inner sep=1.5pt]
\tikzstyle{treenode}=[fill=white, draw=none, shape=circle]

\tikzstyle{arrow}=[->]
\tikzstyle{dashed box}=[-, dashed]
\tikzstyle{condition}=[draw=blue, dashed]

\pgfdeclarelayer{edgelayer}
\pgfdeclarelayer{nodelayer}
\pgfsetlayers{edgelayer,nodelayer,main}
\tikzstyle{none}=[]

\tikzstyle{observation}=[text=blue]
\tikzstyle{morphism}=[fill=white, draw=black, shape=rectangle]
\tikzstyle{medium box}=[fill=white, draw=black, shape=rectangle, minimum width=0.8cm, minimum height=0.9cm]
\tikzstyle{large morphism}=[fill=white, draw=black, shape=rectangle, minimum width=1.7cm, minimum height=1cm]
\tikzstyle{bn}=[fill=black, draw=black, shape=circle, inner sep=1.5pt]
\tikzstyle{state}=[fill=white, draw=black, regular polygon, regular polygon sides=3, minimum width=0.8cm, shape border rotate=90, inner sep=0pt]
\tikzstyle{nn}=[fill=gray, draw=gray, regular polygon, regular polygon sides=3, minimum width=0.4cm, shape border rotate=90, inner sep=0pt]
\tikzstyle{effect}=[fill=white, draw=black, regular polygon, regular polygon sides=3, minimum width=0.8cm, shape border rotate=270, inner sep=0pt]
\tikzstyle{medium state}=[fill=white, draw=black, regular polygon, regular polygon sides=3, minimum width=1.3cm, inner sep=0pt, shape border rotate=90]
\tikzstyle{medium effect}=[fill=white, draw=black, regular polygon, regular polygon sides=3, minimum width=1.3cm, inner sep=0pt, shape border rotate=270]
\tikzstyle{large state}=[fill=white, draw=black, regular polygon, regular polygon sides=3, minimum width=2.2cm, shape border rotate=180, inner sep=0pt]
\tikzstyle{wn}=[fill=white, draw=black, shape=circle, inner sep=1.5pt]
\tikzstyle{treenode}=[fill=white, draw=none, shape=circle]

\tikzstyle{arrow}=[->]
\tikzstyle{dashed box}=[-, dashed]
\tikzstyle{condition}=[draw=blue, dashed]

\tikzset{baseline=(current  bounding  box.center)}
\renewcommand{\tikzfig}[1]{
	\tikzset{baseline=(current  bounding  box.center)}
	\InputIfFileExists{#1.tikz}{}{\input{./graphics/#1.tikz}}
}

\usepackage{color}
\definecolor{deepblue}{rgb}{0,0,0.5}
\definecolor{deepred}{rgb}{0.6,0,0}
\definecolor{deepgreen}{rgb}{0,0.5,0}
\definecolor{darkgray}{rgb}{0.5,0.5,0.5}

\usepackage{listings}

\DeclareFixedFont{\ttb}{T1}{txtt}{bx}{n}{9} 
\DeclareFixedFont{\ttm}{T1}{txtt}{m}{n}{9}  

\lstdefinelanguage{Julia}%
{morekeywords={abstract,break,case,catch,const,continue,do,else,elseif,%
		end,export,false,for,function,immutable,import,importall,if,in,%
		macro,module,otherwise,quote,return,switch,true,try,type,typealias,%
		using,while},%
	sensitive=true,%
	alsoother={$},%
	morecomment=[l]\#,%
	morecomment=[n]{\#=}{=\#},%
	morestring=[s]{"}{"},%
	morestring=[m]{'}{'},%
}[keywords,comments,strings]%

\newcommand*{\mlstinline}[1]{\text{\lstinline|#1|}}
\newcommand*{\code}[1]{\lstinline|#1|}

\newcommand{\cat}{\mathbf}
\newcommand{\catname}{\mathsf}
\newcommand{\C}{\cat C} 
\newcommand{\E}{\mathcal E} 
\newcommand{\F}{\mathcal F} 
\newcommand{\Z}{\mathbb Z} 
\newcommand{\I}{\mathrm{I}}

\newcommand{\id}{\mathsf{id}}

\newcommand{\fib}[1]{\mathbb{#1}}

\newcommand{\R}{\mathbb R}

\newcommand{\N}{\mathcal N}
\renewcommand{\S}{\fib S}

\newcommand{\cpy}{\mathsf{copy}}
\newcommand{\del}{\mathsf{del}}
\newcommand{\mult}{\mathsf{mult}}
\newcommand{\unit}{\mathsf{unit}}

\newcommand{\iv}[1]{\{\!|{#1}|\!\}}

\newcommand{\s}{\,|\,}

\newcommand{\set}{\catname{Set}}

\newcommand{\gauss}{\catname{Gauss}}
\newcommand{\gaussex}{\catname{GaussEx}}

\newcommand{\stoch}{\catname{Stoch}}

\newcommand{\rel}{\catname{Rel}}
\newcommand{\trel}{\catname{TRel}}
\newcommand{\tlinrel}{\catname{TLinRel}}
\newcommand{\vect}{\catname{Vec}}

\newcommand{\meas}{\catname{Meas}}
\newcommand{\chanto}{\leadsto}
\newcommand{\lift}[1]{\langle #1 \rangle}
\renewcommand{\det}{\mathsf{det}}
\renewcommand{\fib}[1]{\mathbb{#1}}

\newcommand{\D}{\fib{D}}

\newcommand{\parto}{\rightharpoonup}
\newcommand{\coparto}{\rightharpoondown}

\newcommand{\Par}{\catname{Par}}
\newcommand{\Copar}{\catname{Copar}}

\newcommand{\im}{\mathsf{im}}

\newcommand{\wquot}{/\!\!/\,}

\newcommand{\hide}[1]{{}}

\newcommand{\B}{\mathcal B}

\begin{document}

\title[A Categorical Treatment of Open Linear Systems]{A Categorical Treatment of Open Linear Systems}


\thanks{It has been helpful to discuss this work with many people. In particular, we would like to thank Paolo Perrone for sharing his insights into point-free measure theory.}

\author[D.~Stein]{Dario Stein\lmcsorcid{0009-0002-1445-4508}}[a]
\author[R.~Samuelson]{Richard Samuelson\lmcsorcid{0009-0006-2330-4404}}[b]

\address{Radboud University Nijmegen, The Netherlands}
\email{dario.stein@ru.nl}
\address{University of Florida, Gainesville, United States}
\email{rsamuelson@ufl.edu}

\begin{abstract}
An \emph{open stochastic system} \`a la Jan Willems is a system affected by two qualitatively different kinds of uncertainty -- one is probabilistic fluctuation, and the other one is nondeterminism caused by a fundamental lack of information. We present a formalization of open stochastic systems in the language of category theory. Central to this is the notion of \emph{copartiality} which models how the lack of information propagates through a system (corresponding to the coarseness of $\sigma$-algebras in Willems' work).

As a concrete example, we study extended Gaussian distributions, which combine Gaussian probability with nondeterminism and correspond precisely to Willems' notion of \emph{Gaussian linear systems}. We describe them both as measure-theoretic and abstract categorical entities, which enables us to rigorously describe a variety of phenomena like noisy physical laws and uninformative priors in Bayesian statistics. The category of extended Gaussian maps can be seen as a mutual generalization of Gaussian probability and linear relations, which connects the literature on categorical probability with ideas from control theory like signal-flow diagrams.
\end{abstract}

\maketitle

\section{Introduction}

When modelling the behavior of systems in computer science and engineering, we encounter different kinds of uncertainty

\begin{enumerate}
\item \textbf{Probabilistic uncertainty}, where we don't know the exact value of some quantity of interest, but we know its statistical distribution. Our knowledge is specified by a \emph{probability distribution} over the set $X$ of possible values of the quantity. A typical example of such a distribution is the \emph{normal} or \emph{Gaussian} distribution $\N(\mu,\sigma^2)$ over the real line, with mean $\mu$ and variance $\sigma^2$.
\item \textbf{Nondeterministic uncertainty}, where the true quantity can be \emph{any} feasible value, but we have no statistical information beyond that. Such uncertainty can be modelled using the \emph{subset} $R \subseteq X$ of the feasible values. Such subsets often arise from equational constraints such as natural laws.
\item \textbf{A lack of information} because the quantity of interest can only be measured incompletely. For example, given two variables $x,y$, we can only measure their difference $x-y$. Mathematically, such lack of information can be formalized using \emph{equivalence relations} $\Phi \subseteq X \times X$, where instead of a quantity $x \in X$, we only have access to its equivalence class $[x]_\Phi \in X/\Phi$. A related notion, which we discuss prominently in this article, is a measurable space $(X,\E)$ with a very coarse $\sigma$-algebra that gives only `fuzzy' access to the elements of $X$.
\end{enumerate}

\noindent Category theory has emerged as a unifying language for studying systems and their composition. Markov categories \cite{fritz} in particular correspond to theories of uncertainty and information flow, and have been used in a variety of applications \cite{fritz:definetti,fritz:zero_one,representable_markov}. In this article, we will give a mathematical analysis of these three kinds of uncertainty in terms of Markov categories. While the semantics of probabilistic and nondeterministic uncertainty are well-known, our model for lack of information in terms of equivalence relations is novel. We propose a notion of \emph{copartial map} that is formally dual to partial maps. In general, copartiality and nondeterminism have different rules of composition (Section~\ref{sec:copartiality}), but we will see that these notions coincide in the case of linear relations (Proposition~\ref{prop:trel_copar}). \\

\noindent \emph{Our main objective} then becomes to model systems under multiple kinds of uncertainty. Examples can be physical systems that obey both relational constraints and feature probabilistic noise, or improper priors in statistical inference.

There is a large body of literature on trying to combine the mathematical theories of probability and nondeterminism into a single well-behaved theory (e.g. \cite{bonchi2020presenting,dahlqvist2018layer,cheung2017distributive,weakdist,varacca2006distributing}). Such approaches are often structured around presenting the relevant theories as monads, and then combining the monads in a suitable way, such as using a distributive law. However it has been shown that no such distributive law can exist in the case of distribution monad (probability) and powerset monad (nondeterminism) \cite{varacca2006distributing,nogo}. A weak distributive law combining these monads does exist \cite{weakdist} but results in a composite monad of \emph{convex subsets of distributions} that no longer has some desirable properties (commutativity). This area is sometimes known as \emph{imprecise probability} and remains an active area of research (e.g. \cite{liell2025compositional}). 

In this work, we take a departure from the monadic approach and focus on combining probability and lack-of-information uncertainty. A framework encompassing both of these forms of uncertainty has been formulated in Willems' \emph{open stochastic systems} \cite{willems:oss}. We will recall what these are, and then work towards a categorical treatment of the matter. 

\paragraph{Open Stochastic Systems}

Our running example will be taken from \cite[Example~1]{willems:oss}.

\begin{exa}[Noisy resistor]\label{ex:resistor}
For an ideal resistor of resistance $R$, Ohm's law states that voltage and current must satisfy the relationship $V=RI$. This is a relational constraint; pairs $(V,I)$ must lie in the subspace $\D = \{ (V,I) : V = RI \}$, but we have no further statistical information about which values the system takes in the allowed subspace. In a realistic system, thermal noise will be present; such a noisy resistor is better modelled by the equation
	\begin{equation}
		V = RI + \epsilon
	\end{equation}
	where $\epsilon \sim \N(0,\sigma^2)$ is a Gaussian random variable with some small variance $\sigma^2$.
\end{exa}

\begin{figure}[h]
	\includegraphics[scale=0.2]{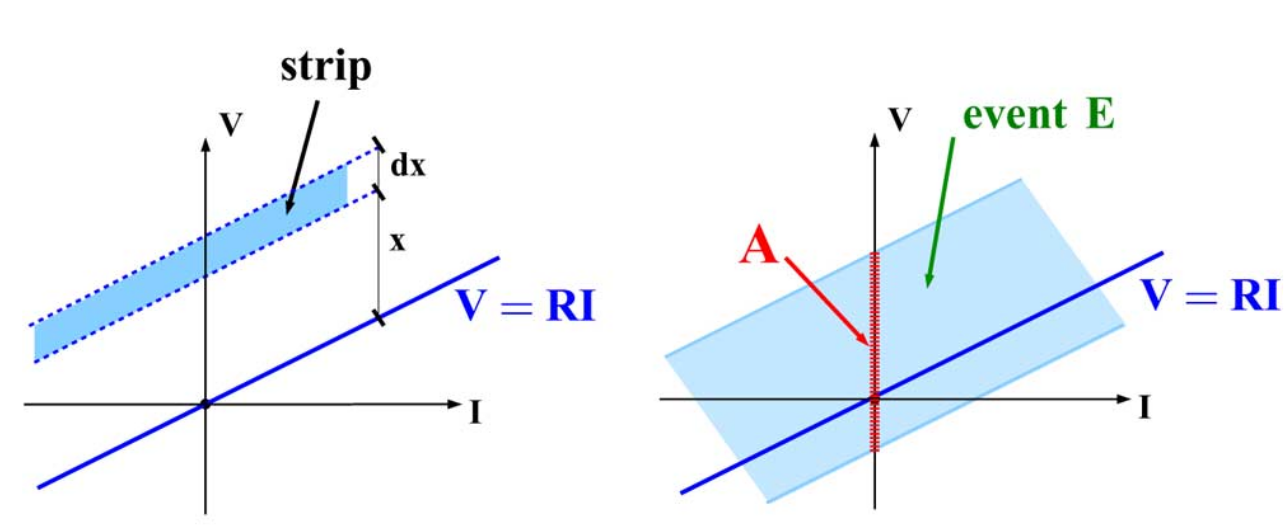}
	\caption{Events for the noisy resistor (Figure 2 of \cite{willems:oss})}
	\label{fig:cylinder}
\end{figure}

\noindent Willems argues that the variables $V$ and $I$ are not random variables in the classical sense; there is no probability distribution associated to them. On the other hand, the quantity $V-RI$ is a honest random variable with a Gaussian distribution. Willems models this situation by endowing the state space $\R^2$ of voltage-current pairs with an unusally coarse $\sigma$-algebra: The only events which are assigned a probability are of the form
$E_A = \{ (V,I) \in \R^2 \s V-RI \in A \}$
where $A \subseteq \R$ is a Borel set. Each set $E_A$ is a Borel cylinder, with sides parallel to the constraint subspace $\D$ from above (see Figure~\ref{fig:cylinder}). The probability of such a cylinder is assigned using the Gaussian distribution of $\epsilon$ as $P(E_A) = \mathsf{Pr}(\epsilon \in A)$.

\noindent Thus, in Willems' approach, lack of information is modelled though coarse $\sigma$-algebras. He speaks of a \emph{closed system} whenever the $\sigma$-algebra contains all Borel sets, giving full access to the underlying quantities. In an \emph{open system}, the $\sigma$-algebra is smaller, indicating a lack of available information. \\

Equational constraints from natural laws are not the only application of open systems. A different use-case arises with uninformative priors in statistical inference:

\begin{exa}[Uninformative Priors]
A typical situation in statistical inference is that we are looking to predict a quantity $X$ while only having access to another quantity $Y$ which is statistically associated with $X$ (for example a noisy measurement, see Figure~\ref{fig:noisymmt}). According to \emph{Bayes' law}, our prediction of $X$ should take into account our prior knowledge about $X$
\[ \underbrace{p(y|x)}_{\text{posterior}} \propto \underbrace{p(x)}_{\text{prior}} \cdot \underbrace{p(y|x)}_{\text{likelihood}} \]
If $(X,Y)$ are jointly Gaussian, then the posterior will be Gaussian, too. This is known as \emph{self-conjugacy} of Gaussians. 

It is unclear under this setup how we should model situations where we completely lack information, either because we have no prior over $X$, or no data about $Y$. We can try approximating the situation by putting larger and larger variances on the variables in question, and see how it affects the predictions in the limit.

However, mathematically, the limit of distributions $\N(0,\sigma^2)$ for $\sigma^2 \to \infty$ approaches the zero measure, which is no longer a probability distribution: There exists no uniform probability distribution over the real line! In practice, one can sometimes pretend (using the method of \emph{improper priors}, e.g. \cite{gelman,hedegaard:gaussian_random_fields}) that $X$ is sampled from the Lebesgue measure (i.e. having constant density $1$). This measure fails to be normalized, however the resulting density calculations may yield the correct posterior. Our theory of extended Gaussians avoids unnormalized measures altogether, as it can express the lack of information as a first-class concept. 
\end{exa}

\begin{figure}
	\includegraphics[width=0.4\textwidth]{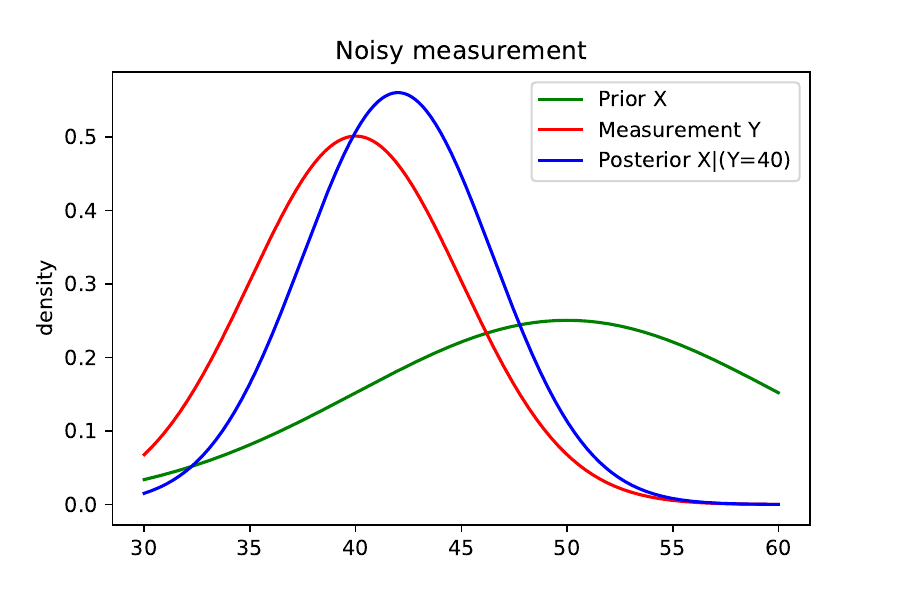}
	\caption{Gaussian prior and posterior in a noisy measurement example}
	\label{fig:noisymmt}
\end{figure}

\paragraph{Compositionality}

Compositionality is the study of combining systems from simpler parts. Willems himself pioneered this discipline with his influential \emph{Behavioral Approach} to control theory \cite{willems:behavioral}. He proposes that complex systems should be understood by ``tearing'' them into simpler parts, ``zooming'' into the parts and modelling their behavior, and then by ``linking'' their behavior together to obtain the composite behavior. \emph{Monoidal category theory} is a widely used language for compositional processes that has been applied to numerous settings in mathematics, computer science, statistics, physics and engineering (e.g. \cite{fong,fritz,baez:control,baez:props,bonchi:cat_signal_flow}). For a general overview, see for example \cite{fong:sevensketches}. 

A \emph{category} features objects $X,Y,\ldots$, thought of as state spaces, and ``morphisms'' $X \to Y$ between these objects, representing abstract processes. Note that while morphisms are oriented to have input $X$ and output $Y$, this does not mean that information flows unidirectionally from $X$ to $Y$. A \emph{state} is a morphism $\I \to X$ without input (formally, its input is a special unit object $\I$). Morphisms and can be composed in sequence and in parallel. We can draw such composites using the diagrammatic language of string diagrams \cite{selinger:graphical}.

\[ \begin{tikzpicture}[scale=0.3]
	\begin{pgfonlayer}{nodelayer}
		\node [style=none] (0) at (-4, 2) {};
		\node [style=morphism] (1) at (-2, 2) {$f$};
		\node [style=morphism] (2) at (2, 2) {$g$};
		\node [style=none] (3) at (4, 2) {};
		\node [style=none] (4) at (0, -2) {Sequential composition};
		\node [style=none] (5) at (-5, 2) {$X$};
		\node [style=none] (6) at (5, 2) {$Z$};
		\node [style=none] (7) at (0, 3) {$Y$};
		\node [style=none] (8) at (13, 3.5) {};
		\node [style=morphism] (9) at (15, 3.5) {$f_1$};
		\node [style=none] (10) at (17, 3.5) {};
		\node [style=none] (11) at (13, 0.5) {};
		\node [style=morphism] (12) at (15, 0.5) {$f_2$};
		\node [style=none] (13) at (17, 0.5) {};
		\node [style=none] (14) at (12, 3.5) {$X_1$};
		\node [style=none] (15) at (12, 0.5) {$X_2$};
		\node [style=none] (16) at (18, 3.5) {$Y_1$};
		\node [style=none] (17) at (18, 0.5) {$Y_2$};
		\node [style=none] (18) at (15, -2) {Parallel composition};
		\node [style=state] (19) at (-15, 3.5) {$\psi$};
		\node [style=none] (20) at (-12, 3.5) {};
		\node [style=none] (21) at (-11, 3.5) {$X$};
		\node [style=effect] (22) at (-12, 0.5) {$\epsilon$};
		\node [style=none] (23) at (-15, 0.5) {};
		\node [style=none] (24) at (-15, 0.5) {};
		\node [style=none] (25) at (-16, 0.5) {$X$};
		\node [style=none] (26) at (-13.25, -2) {States \& effects};
	\end{pgfonlayer}
	\begin{pgfonlayer}{edgelayer}
		\draw (0.center) to (1);
		\draw (1) to (2);
		\draw (2) to (3.center);
		\draw (8.center) to (9);
		\draw (9) to (10.center);
		\draw (11.center) to (12);
		\draw (12) to (13.center);
		\draw (19) to (20.center);
		\draw (22) to (24.center);
	\end{pgfonlayer}
\end{tikzpicture} \]

\noindent Categories of processes often have richer structure which lets morphisms be interconnected in more complex ways. Examples of such structures are copy-delete categories and hypergraph categories, which allow information to be shared, discarded or compared.  

\[ \begin{tikzpicture}[scale=0.3]
	\begin{pgfonlayer}{nodelayer}
		\node [style=none] (0) at (-14, 0) {};
		\node [style=bn] (1) at (-12, 0) {};
		\node [style=none] (2) at (-10, 1) {};
		\node [style=none] (3) at (-10, -1) {};
		\node [style=none] (4) at (-3.75, 0) {};
		\node [style=bn] (5) at (-1.75, 0) {};
		\node [style=none] (6) at (15.75, 0) {};
		\node [style=bn] (7) at (14, 0) {};
		\node [style=none] (8) at (12, 1) {};
		\node [style=none] (9) at (12, -1) {};
		\node [style=none] (10) at (4, 0) {};
		\node [style=bn] (11) at (2, 0) {};
		\node [style=none] (12) at (-12, -4) {copy};
		\node [style=none] (13) at (-3, -4) {discard};
		\node [style=none] (14) at (3, -4) {initialize};
		\node [style=none] (15) at (14, -4) {interconnect};
		\node [style=none] (16) at (-15, 0) {$X$};
		\node [style=none] (17) at (-9, 1) {$X$};
		\node [style=none] (18) at (-9, -1) {$X$};
		\node [style=none] (19) at (-4.75, 0) {$X$};
		\node [style=none] (20) at (5, 0) {$X$};
		\node [style=none] (21) at (11, 1) {$X$};
		\node [style=none] (22) at (16.75, 0) {$X$};
		\node [style=none] (23) at (11, -1) {$X$};
	\end{pgfonlayer}
	\begin{pgfonlayer}{edgelayer}
		\draw (0.center) to (1);
		\draw [in=-180, out=90, looseness=1.25] (1) to (2.center);
		\draw [in=-90, out=-180, looseness=1.25] (3.center) to (1);
		\draw (4.center) to (5);
		\draw (6.center) to (7);
		\draw [in=0, out=90, looseness=1.25] (7) to (8.center);
		\draw [in=-90, out=0, looseness=1.25] (9.center) to (7);
		\draw (10.center) to (11);
	\end{pgfonlayer}
\end{tikzpicture} \]

\noindent Willems' compositional approach and the categorical language fit well together. For example, aspects of Willems' work on relational constraints has been categorified in \cite{fong:open_systems}. We wish to extend this approach to systems with relational and stochastic constraints, like the noisy resistor. However, we run into two immediate difficulties:

\begin{enumerate}
\item The first concerns the use of the word ``open''. To Willems, this means a lack of information, modelled by a coarse $\sigma$-algebra. In category theory, ``open'' generally refers to morphisms $X \to Y$ with inputs, as opposed to mere states $\I \to Y$. We will investigate how these notions are related, and can be unified in a single framework.
\item Secondly, in measure-theoretic probability, the objects are usually measurable spaces $(X,\E)$. That is, the $\sigma$-algebra $\E$ is treated as fixed, because it is part of the objects. In open stochastic systems on the other hand, the $\sigma$-algebra varies, as different computations will result in different $\sigma$-algebras. We can thus no longer treat $\sigma$-algebras as part of the objects, but must make them part of the morphisms. This leads us to formalize the phenomenon of \emph{copartiality}. 
\end{enumerate}

\noindent To elaborate the second point, let us revisit the noisy resistor. We would like model its behavior as a joint state $P_{VI} : \I \to \R^2$. Composing with the projections $\R^2 \to \R$ should tell us the marginal distributions of $V,I$, as shown below
\[\begin{tikzcd}
	&&&& \R \\
	\I && {\R^2} \\
	&&&& \R
	\arrow["{P_{VI}}", from=2-1, to=2-3]
	\arrow["{P_V}", curve={height=-12pt}, from=2-1, to=1-5]
	\arrow["{\pi_2}"', from=2-3, to=3-5]
	\arrow["{\pi_1}", from=2-3, to=1-5]
	\arrow["{P_I}"', curve={height=12pt}, from=2-1, to=3-5]
\end{tikzcd}\]

\noindent However, we won't know the $\sigma$-algebra on the codomain $\R$ until we compute the composite. In the case of $I,V$, the $\sigma$-algebra turns out to be indiscrete, i.e. $\{\emptyset,\R\}$, which means that we have no statistical information about them. If instead we consider the marginal of the variable $V-RI$, we obtain a Gaussian distribution on the ordinary Borel $\sigma$-algebra. \\

\noindent We address these challenges as follows: We first introduce a notion of \textbf{extended Gaussian distribution}, which precisely corresponds to Willems' Gaussian linear systems. We then define a Markov category of \textbf{extended Gaussian maps} which formalizes how these distributions compose. 

Roughly speaking, an extended Gaussian distribution is a pair $(\psi,\D)$ of a subspace $\D \subseteq \R^n$, and a Gaussian distribution on the quotient space $\R^n/\D$. We may write this in additive notation $\psi + \D$, the same way we write elements of quotients as cosets, e.g. $(1 + 2\Z) \in \Z/2\Z$. The subspace $\D$ is called the fibre of the distribution, and corresponds to the nondeterministic contribution. If $\D = 0$, we are talking about ordinary Gaussian distributions. 

For example, the noisy resistor can be described by the extended Gaussian distribution
\[ P_{VI} = \N\left(
\begin{pmatrix} \sigma^2 \\ 0 \end{pmatrix},
\begin{pmatrix} 0 & 0 \\ 0 & 0 \end{pmatrix}\right) 
+ \left\{ \begin{pmatrix} V \\ I \end{pmatrix} \s V = RI \right\}
\]

\noindent We can straightforwardly compute $P_V = P_I = \R$; this means that they can take arbitrary values on the real line but no statistical distribution is known. On the other hand, computing the distribution of $V-RI$ returns, as expected, the Gaussian distribution $\N(0,\sigma^2)$. 

\subsection{Contribution}

We answer the following questions: How can Willems' framework of open stochastic systems be captured with category theory, and how does his notion of interconnection be reconciled with categorical composition? \\

We demonstrate this by defining a category $\gaussex$ of extended Gaussian maps, which faithfully combines Gaussian probability with nondeterminism. The states of $\gaussex$ correspond precisely to \emph{Gaussian linear systems} of Willems. We describe these entities both in measure-theoretic terms, as well as through abstract entities obeying certain composition laws (without reference to measure theory).

We introduce \emph{copartiality} as a way to formalize lack of information via maps into quotients. For linear maps, copartiality and relational nondeterminism coincide, which allows us to see  $\gaussex$ as a mutual generalization of Gaussian probability and linear relations. Then, we return to Willems' measure-theoretic approach and prove that categorical composition and his notion of interconnection coincide (Theorem~\ref{thm:interconnection_formula}). \\

\noindent \emph{Note:} This article is based on our CALCO paper \cite{stein2023category}, where we first introduce extended Gaussian distributions. In this invited journal submission, we expand on this previous work with greater detail and place it in a larger context. In particular, we have rewritten the development to follow Willems' measure-theoretic definitions alongside our categorical ones. This allows us to prove a novel theorem that these definitions agree in Theorem~\ref{thm:interconnection_formula}. We also give a standalone introduction to copartiality in Section~\ref{sec:copartiality}, and elaborate the treatment of quotients and equivalence of systems in Section~\ref{sec:wquot}. Since the publication of \cite{stein2023category}, two follow-up publications \cite{stein2023towards,stein2024graphical} have established connections between extended Gaussian distributions and convex optimization. We elaborate these connection in Section~\ref{sec:connections} but keep this work self-contained until that point.

\subsection{Outline}

In Section~\ref{sec:oss}, we recall the basics of measure theory and Gaussian probability, before defining Gaussian linear systems/extended Gaussian distributions in Section~\ref{sec:gaussex_distributions}.

In Section~\ref{sec:ct}, we discuss the relevant notions from monoidal category theory, in particular Markov categories, and give example categories which model relational nondeterminism (\ref{sec:relations}), measure-theoretic probability (\ref{sec:channels}), and Gaussian probability (\ref{sec:gauss}).

Section~\ref{sec:gaussex} features our main novel definitions, namely the Markov category $\gaussex$. We begin by introducing copartiality (\ref{sec:copartiality}) and the relationship with linear relations (\ref{sec:linrel}), before defining $\gaussex$ using a decorated cospan construction (\ref{sec:gaussex_def}).

In Section~\ref{sec:interconnection}, we connect our categorical notion of composition back to Willems' notion of interconnection of complementary systems. Our main theorem~\ref{thm:interconnection_formula} shows that these agree. The section ends with \ref{sec:conditioning} explaining how interconnection can be generalized beyond complementary systems. 

In Section~\ref{sec:connections}, we give some wider context on extended Gaussian distributions and the details of our construction. \ref{sec:wquot} concerns the question of equivalence of systems and the relationship between coarse $\sigma$-algebras and quotients, while \ref{sec:duality} and \ref{sec:presentations} discuss recent work on the duality theory of extended Gaussian distributions and a complete axiomatization for them. We finish with a proof of concept implementation of this work as a Julia library Section~\ref{sec:implementation}. 

\section{Open Stochastic Systems}\label{sec:oss}

In this section, we introduce Gaussian linear systems in the sense of Willems, and our abstract description of them, which we term an \emph{extended Gausssian distribution}. We begin by briefly recalling the measure-theoretic foundations of probability. For a standard reference, see \cite{kallenberg}.

\emph{Notation:} We write $f[A]$ for the image of a subset $A \subseteq X$ under a function $f : X \to Y$. We use the Iverson brackets for indicator functions, i.e. the expression $[\phi]$ evaluates to $1$ if $\phi$ is true, and $0$ otherwise.

A \emph{$\sigma$-algebra} on a set $X$ is a collection $\E$ of subsets of $X$ which contains $\emptyset$ and is closed under countable unions and complementation. For topological spaces like $\R^n$, we will consider the Borel $\sigma$-algebra $\B(\R^n)$, which is the smallest $\sigma$-algebra containing all open subsets. 

A \emph{measurable space} is a pair $(X,\E)$ of a set and a $\sigma$-algebra on it. A function $f : (X,\E) \to (Y,\F)$ between measurable spaces is called \emph{measurable} if $f^{-1}(F) \in \E$ for all $F \in \F$. We will refer to spaces like $(\R^n, \B(\R^n))$ equipped with their Borel $\sigma$-algebra as \emph{Borel spaces}. The product of two measurable spaces $(X,\E)$ and $(Y,\F)$ is the set $X \times Y$ equipped with the product $\sigma$-algebra $\E \otimes \F$, that is the $\sigma$-algebra generated by the measurable rectangles $E \times F$ for $E \in \E, F \in \F$. The product of two Borel spaces is again Borel. Measurable spaces and measurable functions form a category $\meas$.

A \emph{probability measure} or \emph{distribution} on $(X,\E)$ is a function $P : \E \to [0,1]$ which satisfies $P(X)=1$ and
\[ P\left(\bigcup_{i=1}^\infty E_i\right) = \sum_{i=1}^\infty P(E_i)\]
for all disjoint collections $(E_i)_{i=1}^\infty$ of sets in $\E$. An important example is the Dirac distribution $\delta_x$ at point $x \in X$, which is defined by $\delta_x(E) = [x \in E]$.

For a probability measure $P$ and a measurable function $f : (X,\E) \to (\R, \B(\R))$, the integral is denoted by $\int f(x) P(\mathrm dx)$. If $f : (X,\E) \to (Y,\F)$ is a measurable function and $P$ is a probability measure on $(X,\E)$, we can define a \emph{pushforward measure} $f_*P$ on $(Y,\F)$ by $f_*P(F) = P(f^{-1}(F))$. For two distributions $P_1,P_2$ on $(X_1,\E_1), (X_2,\E_2)$ respectively, their product distribution $P_1 \otimes P_2$ on $(P_1 \times P_2, \E_1 \otimes \E_2)$ is defined on measurable rectangles by $(P_1 \otimes P_2)(E_1 \times E_2) = P_1(E_1) P_2(E_2)$ and uniquely extended to a probability measure to the whole product $\sigma$-algebra. 

A \emph{probability space} is a tuple $(X,\E,P)$ of a measurable space equipped with a probability measure. In this setting, the measurable sets $E \in \E$ are also called \emph{events}. In Willems' terminology, a probability space is called a \emph{stochastic system}, and we will use both names interchangeably.

\subsection{Gaussian distributions} The class of probability distributions we will be concerned with in this article are the \emph{Gaussian distributions}, also known as multivariate normal distributions. A univariate Gaussian distribution on $\R$ has the form 
\[ P(A) = \int_A \frac{1}{\sqrt{2\pi\sigma^2}} \exp\left(-\frac{(x-\mu)^2}{2\sigma^2}\right) \mathrm dx \] 
for numbers $\mu,\sigma^2$, which are the mean and variance of $P$. We denote this probability measure as $\N(\mu,\sigma^2)$. Note that the Dirac distribution is considered a special case of a Gaussian distribution with variance $0$, i.e. $\delta_\mu = \N(\mu,0)$.

Similarly, a multivariate Gaussian distribution on $\R^n$ is fully characterized by its mean $\mu \in \R^n$ and its covariance matrix, which is a positive-semidefinite matrix $\Sigma \in \R^{n \times n}$. We write this distribution as $\N(\mu,\Sigma)$, and define $\gauss(\R^n)$ to be the set of all Gaussian distributions on $\R^n$.

\begin{propC}[{\cite[Section~6]{fritz}}]\label{prop:gauss_laws}
Product distributions, convolutions and pushforwards of Gaussian distributions under linear maps (matrices) $M$ are Gaussian, and satisfy	the laws
\begin{align*}
M_*(\N(\mu,\Sigma)) &= \N(M\mu,M\Sigma M^T) \\
\N(\mu_1,\Sigma_1)\otimes \N(\mu_2,\Sigma_2) &= \N\left(
\begin{pmatrix} \mu_1 \\ \mu_2 \end{pmatrix},
\begin{pmatrix} \Sigma_1 & 0 \\ 0 & \Sigma_2 \end{pmatrix}
\right) \\
\N(\mu_1,\Sigma_1) + \N(\mu_2,\Sigma_2) &= \N(\mu_1 + \mu_2, \Sigma_1 + \Sigma_2)
\end{align*}
\end{propC}

This means that we can often manipulate Gaussian distributions without resorting to measure-theoretic computation. Instead, we can treat them as abstract, formal entities that obey certain transformation rules. We will leverage this in our categorical presentation, and develop a similar, abstract presentation of extended Gaussian distributions in Section~\ref{sec:gaussex_def}. 

\subsection{Open Linear Systems}\label{sec:gaussex_distributions}

We now come to the central notion of Willems. Let $\fib D \subseteq \R^n$ be a vector subspace. A \emph{Borel cylinder} parallel to $\fib D$ is a Borel set $A \subseteq \R^n$ which is closed under translations along $\fib D$, i.e. $A + \D \subseteq A$. The Borel cylinders define a sub-$\sigma$-algebra
\[ \B_{\fib D}(\R^n)= \{ A \in \B(\R^n) \s A + \D \subseteq A \}\]
of the Borel sets. 

\begin{defi}
We write $\R^n \wquot \D$ for the measurable space $(\R^n, \B_{\D}(\R^n))$. We call a stochastic system \emph{linear} if its underlying measurable space is $\R^n \wquot \D$, and call $\D$ the fibre of the system. We call a linear system \emph{closed} if its fibre vanishes; in this case, its $\sigma$-algebra is all of $\B(\R^n)$. Otherwise, the $\sigma$-algebra is strictly smaller, and we call the system \emph{open}.
\end{defi}

The space $\R^n \wquot \D$ has similarities to the quotient space $\R^n/\D$. We spell out their precise relationship in Section~\ref{sec:wquot}. In an open linear system, only events parallel to $\D$ are assigned a probability. Openness thus corresponds to a lack of information about the system. 

\begin{defi}
An \emph{extended Gaussian distribution} $\chi$ on $\R^n$ is the restriction of a Gaussian probability measure on $\R^n$ to the Borel cylinders $\B_\D(\R^n)$. We call $\D$ the fibre of the distribution.
A \emph{Gaussian system} is a space $\R^n$ equipped with an extended Gaussian distribution. We write $\gaussex(\R^n)$ for the set of extended Gaussian distributions on $\R^n$.
\end{defi}

It is important that the $\sigma$-algebra (or equivalently) the fibre is part of the specification of an extended Gaussian distributions. Different distributions $\chi_1,\chi_2$ on $\R^n$ will in general be defined on different $\sigma$-algebras; the underlying measurable space of a Gaussian system with fibre $\D$ is $\R^n \wquot \D$.

\begin{exa}
Every Gaussian distribution on $\R^n$ is also a an extended Gaussian distribution, with fibre $\D=0$.
\end{exa}

\begin{exa}
Every subspace $\D \subseteq \R^n$ can be seen as an extended Gaussian distribution with fibre $\D$, namely the restriction of $\delta_0$ to $\B_\D(\R^n)$. Explicitly, we have
\[ \D : \B_\D(\R^n) \to [0,1], \quad \D(A) = [0 \in A] \]
We think of $\D$ as an idealized uniform distribution. It indicates an arbitrary point in the subspace $\D$, about which we have no further information. 
\end{exa}

We will see at the end of this section that the general extended Gaussian distribution is a sum (convolution) of an ordinary Gaussian distribution and a subspace. 

\paragraph{Kernel representation}

Our chief example of extended Gaussian distribution arises from pulling back Gaussian distributions under a surjective map. 

\begin{defi}
Let $p : (X,\E) \to (Y,\F)$ be a surjective measurable function, and $\psi : \F \to [0,1]$ a distribution on $Y$. Then we can form the pullback $\sigma$-algebra
\[ p^{-1}[\F] = \{ p^{-1}(F) : F\in \F \}\]
and define a measure $p^{-1}(\psi) : p^{-1}[\F] \to [0,1]$ as
\[ p^{-1}(\psi)(p^{-1}(F)) = \psi(F) \]
for all $F \in \F$. 
Measurability ensures that $p^{-1}[\F] \subseteq \E$ is a sub-$\sigma$-algebra. 
\end{defi}

\noindent We will revisit the details of the pullback construction in Section~\ref{sec:wquot}.

\begin{prop}\label{prop:kernel_rep}
For every surjective linear map $p : \R^n \to \R^k$ and Gaussian distribution $\psi \in \gauss(\R^k)$, we have that the pullback distribution $p^{-1}(\psi)$ is an extended Gaussian distribution on $\R^n$ with fibre $\D = \ker(p)$. Every extended Gaussian distribution $\chi$ arises this way, so we call $p^{-1}(\psi)$ a \emph{kernel representation} of $\chi$.
\end{prop}
\begin{proof}
Choose a complementary subspace $\fib K$ with $\fib K \oplus \D = \R^n$. Then there exists an injective linear map $\iota : \R^k \to \R^n$ whose image equals $\fib K$; this defines a section $p\iota = \id_{\R^k}$. 
	
The pullback $\sigma$-algebra $\E = p^{-1}[\B(\R^k)]$ consists precisely of the Borel cylinders parallel to $\D = \ker(p)$: Every cylinder $C$ is of the form $p^{-1}(A) = i[A] + \D$ for a unique Borel subset $A \subseteq \R^k$. 

We claim that $\chi = p^{-1}(\psi)$ is equal to the restriction of the Gaussian distribution $\iota_*\psi$ to $\B_\D(\R^n)$. Indeed, for an arbitrary cylinder $C=i[A] + \D$ we check $\chi(C) = \psi(A) = \psi(\iota^{-1}(C)) = (\iota_*\psi)(C)$. 

Conversely, let $\psi$ be a Gaussian distribution on $\R^n$, $\D \subseteq \R^n$ and $\fib K$ a complement. Let $p$ be the composite projection $\R^n \xrightarrow{p_\fib K} \fib K \cong \R^k$. Then the restriction $\psi|_{\B_\D(\R^n)}$ is equal to the pullback $p^{-1}(p_*\psi)$. 
\end{proof}

\begin{prop}
	For surjective linear maps $p_i : \R^n \to \R^{k_i}$ and distributions $\psi_i \in \gauss(\R^{k_i})$ with $i=1,2$, we have $p_i^{-1}(\psi_1) = p_2^{-1}(\psi_2)$ if and only if there exists a linear isomorphism $\tau : \R^{k_1} \to \R^{k_2}$ with 
	\[\begin{tikzcd}
		{\R^{k_1}} \\
		&& {\R^n} \\
		{\R^{k_2}}
		\arrow["{p_1}"', two heads, from=2-3, to=1-1]
		\arrow["{p_2}", two heads, from=2-3, to=3-1]
		\arrow["\tau"', from=1-1, to=3-1]
	\end{tikzcd}  \qquad \text{ with } \tau_*\psi_1 = \psi_2 \]
\end{prop}
\begin{proof}[Proof sketch]
``$\Rightarrow$'' If $p_1^{-1}(\psi_1) = p_2^{-1}(\psi_2)$ then in particular $\ker(p_1) = \ker(p_2)$ because the pullback distributions have the same $\sigma$-algebra and hence fibre $\D$. In this case, there exists a unique isomorphism $\tau$ with $\tau \circ p_1 = p_2$ from the rank-nullity theorem $\R^{k_1} \cong \R^n/\D \cong \R^{k_2}$. The isomorphism must preserve the distribution $\tau_*\psi_1 = \psi_2$ as equality of the pullback distributions can be checked explicitly on a complement $\mathbb K$ of $\D$ (see proof of Proposition~\ref{prop:kernel_rep}).
``$\Leftarrow$'' Given the distribution preserving isomorphism $\tau$, we immediately derive 
\[ p_2^{-1}(\psi_2) = (\tau \circ p_1)^{-1}(\tau_*\psi_1) = p_1^{-1}(\tau^{-1}(\tau_*\psi_1)) = p_1^{-1}(\psi_1) \qedhere \]
\end{proof}

\begin{exa}\label{ex:noisy_kernel}
The noisy resistor from the introduction can be described as the extended Gaussian distribution on $\R^2$ with kernel representation $p^{-1}(\N(0,\sigma^2))$, where $p(V,I) = V-RI$.
\end{exa}

\paragraph{Arithmetic of Quotients}

We wish to give transformation rules of extended Gaussian distributions the same way as in Proposition~\ref{prop:gauss_laws}. The complication in, say, trying to compute the convolution of two extended Gaussian distributions $\chi_1, \chi_2$ on $\R^n$, is that they are not necessarily defined on the same $\sigma$-algebra if their fibres differ. This means that $\sigma$-algebras (equivalently fibres) need to become part of the calculation. 

This is analogous to asking what the sum of two equivalence classes $[x_1]\in \R^n/\D_1$ and $[x_2] \in \R^n/\D_2$ should be. We propose to answer this by $[x_1+x_2] \in \R^n/(\D_1 + \D_2)$. This is intuitive if we think of equivalence classes as cosets, and simplify according to the usual arithmetic $(x_1 + \D_1) + (x_2 + \D_2) = (x_1 + x_2) + (\D_1 + \D_2)$. 

We will from now on denote the restriction of $\psi$ to $\B_\D(\R^n)$ as $\psi + \D$. This can be understood as mere notation, but it is more than that, namely the convolution of $\psi$ and $\D$ both seen as extended Gaussian distributions. 

\begin{defi}\label{def:gaussex_state_rules}
We define the following transformation rules for extended Gaussian distributions under linear maps (matrices) $M$
	\begin{align*}
		M_*(\psi + \D) &= M_*\psi + M[\D] \\
		(\psi_1 + \D_1) \otimes (\psi_2 + \D_2) &= \psi_1 \otimes \psi_2 + \D_1 \times \D_2 \\
		(\psi_1 + \D_1) + (\psi_2 + \D_2) &= (\psi_1 + \psi_2) + (\D_1 + \D_2) 
	\end{align*}
	Here $M[\D]$ denotes the image subspace of $\D$ under $M$, $\D_1 \times \D_2$ is cartesian product and $\D_1 + \D_2$ is Minkowski sum.
\end{defi}

At this point, it is unclear if these rules obey the expected laws, or if they are even well-defined. We will answer this question affirmatively in Proposition~\ref{prop:gaussex_state_rules} by constructing a category $\gaussex$ of extended Gaussian maps in which extended Gaussian distributions correspond to the states. This category will be a Markov category, which is one common axiomatization of what it means to be a categorical model of probability theory. Furthermore $\gaussex$ faithfully contains both Gaussian distributions and linear relations, showing that the composition rules for extended Gaussians conservatively extend both of these examples.  

\section{Overview of Categorical Probability}\label{sec:ct}

We now want to study how extended Gaussian distributions can be combined and transformed. The crucial step is to extend our focus from extended Gaussian \emph{distributions} on a space $X$, to a notion of map between spaces $X \to Y$. Extended Gaussian distributions are then recovered as maps $I \to X$ with trivial inputs. This category should obey certain laws of categorical probability theory, as formalized by the axioms of a Markov category \cite{fritz}. \\

Recall that a monoidal category $(\C,\otimes,I)$ comprises objects $X,Y,\ldots$ and morphisms $f : X \to Y$ between these objects. Morphisms can be composed in two ways, namely in sequence and in parallel:
\begin{enumerate}
\item if $f : X \to Y$ and $g : Y \to Z$, we have a sequential composite $g \circ f : X \to Z$. This is also denoted $gf$ or $f ; g$ (diagrammatic composition order).
\item if $f_i : X_i \to Y_i$ for $i=1,2$, then we have a parallel composite $f_1 \otimes f_2 : X_1 \otimes X_2 \to Y_1 \otimes Y_2$
\end{enumerate}

\begin{defiC}[{\cite{cho2019disintegration}}]
A \emph{copy-delete} category is a symmetric monoidal category $(\C,\otimes,I)$ where every object is equipped with the structure of a commutative comonoid 
\[ \cpy_X = 
	\tikzset{baseline=-0.65ex}
	\begin{tikzpicture}[scale=0.1]
	\begin{pgfonlayer}{nodelayer}
		\node [style=bn] (1) at (0, 0) {};
		\node [style=none] (3) at (-3, 0) {};
		\node [style=none] (4) at (2, 2) {};
		\node [style=none] (5) at (2, -2) {};
		\node [style=none] (6) at (4, -2) {};
		\node [style=none] (7) at (4, 2) {};
	\end{pgfonlayer}
	\begin{pgfonlayer}{edgelayer}
		\draw (1) to (3.center);
		\draw [in=-180, out=90] (1) to (4.center);
		\draw (4.center) to (7.center);
		\draw [in=-180, out=-90] (1) to (5.center);
		\draw (5.center) to (6.center);
	\end{pgfonlayer}
\end{tikzpicture}}
 : X \to X \otimes X, \quad \del_X = 
	\tikzset{baseline=-0.65ex}
	\begin{tikzpicture}[scale=0.1]
	\begin{pgfonlayer}{nodelayer}
		\node [style=bn] (1) at (2, 0) {};
		\node [style=none] (3) at (-1, 0) {};
	\end{pgfonlayer}
	\begin{pgfonlayer}{edgelayer}
		\draw (1) to (3.center);
	\end{pgfonlayer}
\end{tikzpicture}
}
 : X \to I \]
subject to axioms, rendered graphically as follows,
\[ \begin{tikzpicture}[scale=0.2]
	\begin{pgfonlayer}{nodelayer}
		\node [style=none] (0) at (-2, 0) {};
		\node [style=bn] (1) at (0, 0) {};
		\node [style=bn] (2) at (2, 2) {};
		\node [style=none] (3) at (2, -2) {};
		\node [style=none] (4) at (4, 1) {};
		\node [style=none] (5) at (4, 3) {};
		\node [style=none] (6) at (4, -2) {};
		\node [style=none] (7) at (6, 0) {=};
		\node [style=none] (8) at (8, 0) {};
		\node [style=bn] (9) at (10, 0) {};
		\node [style=bn] (10) at (12, -2) {};
		\node [style=none] (12) at (14, -3) {};
		\node [style=none] (13) at (14, -1) {};
		\node [style=none] (14) at (12, 2) {};
		\node [style=none] (15) at (14, 2) {};
		\node [style=none] (16) at (19, 0) {};
		\node [style=bn] (17) at (21, 0) {};
		\node [style=bn] (18) at (23, -2) {};
		\node [style=none] (21) at (23, 2) {};
		\node [style=none] (22) at (25, 2) {};
		\node [style=none] (23) at (28, 0) {};
		\node [style=bn] (24) at (30, 0) {};
		\node [style=bn] (25) at (32, 2) {};
		\node [style=none] (26) at (32, -2) {};
		\node [style=none] (29) at (34, -2) {};
		\node [style=none] (30) at (26, 0) {=};
		\node [style=none] (31) at (35, 0) {=};
		\node [style=none] (32) at (36.5, 0) {};
		\node [style=none] (33) at (41.5, 0) {};
		\node [style=none] (34) at (46, 0) {};
		\node [style=bn] (35) at (48, 0) {};
		\node [style=none] (36) at (50, 2) {};
		\node [style=none] (37) at (50, -2) {};
		\node [style=none] (38) at (53, -2) {};
		\node [style=none] (39) at (53, 2) {};
		\node [style=none] (40) at (54, 2) {};
		\node [style=none] (41) at (54, -2) {};
		\node [style=none] (42) at (55, 0) {=};
		\node [style=none] (43) at (57, 0) {};
		\node [style=bn] (44) at (59, 0) {};
		\node [style=none] (45) at (61, 2) {};
		\node [style=none] (46) at (61, -2) {};
		\node [style=none] (47) at (15.5, -0.25) {,};
		\node [style=none] (48) at (43, -0.25) {,};
		\node [style=none] (49) at (0, -10) {};
		\node [style=bn] (50) at (4, -10) {};
		\node [style=none] (52) at (2, -8) {$X \otimes Y$};
		\node [style=none] (53) at (6, -10) {=};
		\node [style=none] (54) at (8, -8.5) {};
		\node [style=bn] (55) at (12, -8.5) {};
		\node [style=none] (56) at (8, -11.5) {};
		\node [style=bn] (57) at (12, -11.5) {};
		\node [style=none] (58) at (10, -7) {$X$};
		\node [style=none] (59) at (10, -13) {$Y$};
		\node [style=none] (60) at (46, -10) {};
		\node [style=bn] (61) at (50, -10) {};
		\node [style=none] (62) at (52, -8) {};
		\node [style=none] (63) at (52, -12) {};
		\node [style=none] (64) at (55, -10) {=};
		\node [style=none] (66) at (47.5, -8) {$X \otimes Y$};
		\node [style=none] (67) at (56, -9) {};
		\node [style=bn] (68) at (59, -9) {};
		\node [style=none] (69) at (61, -7) {};
		\node [style=none] (70) at (61, -11) {};
		\node [style=none] (71) at (56, -11) {};
		\node [style=bn] (72) at (59, -11) {};
		\node [style=none] (73) at (61, -9) {};
		\node [style=none] (74) at (61, -13) {};
		\node [style=none] (75) at (57.5, -7) {$X$};
		\node [style=none] (76) at (57.5, -13) {$Y$};
	\end{pgfonlayer}
	\begin{pgfonlayer}{edgelayer}
		\draw (0.center) to (1);
		\draw [in=-180, out=90] (1) to (2);
		\draw [in=180, out=-90] (1) to (3.center);
		\draw [in=-180, out=90, looseness=1.25] (2) to (5.center);
		\draw [in=180, out=-90, looseness=1.25] (2) to (4.center);
		\draw (3.center) to (6.center);
		\draw (8.center) to (9);
		\draw [in=-180, out=90, looseness=1.25] (10) to (13.center);
		\draw [in=180, out=-90, looseness=1.25] (10) to (12.center);
		\draw (15.center) to (14.center);
		\draw [in=90, out=180] (14.center) to (9);
		\draw [in=-180, out=-90] (9) to (10);
		\draw (16.center) to (17);
		\draw (22.center) to (21.center);
		\draw [in=90, out=180] (21.center) to (17);
		\draw [in=-180, out=-90] (17) to (18);
		\draw (23.center) to (24);
		\draw [in=-180, out=90] (24) to (25);
		\draw [in=180, out=-90] (24) to (26.center);
		\draw (26.center) to (29.center);
		\draw (33.center) to (32.center);
		\draw (34.center) to (35);
		\draw [in=-180, out=90] (35) to (36.center);
		\draw [in=180, out=-90] (35) to (37.center);
		\draw [in=180, out=0, looseness=0.50] (36.center) to (38.center);
		\draw [in=-180, out=0, looseness=0.50] (37.center) to (39.center);
		\draw (41.center) to (38.center);
		\draw (40.center) to (39.center);
		\draw (43.center) to (44);
		\draw [in=-180, out=90] (44) to (45.center);
		\draw [in=180, out=-90] (44) to (46.center);
		\draw (50) to (49.center);
		\draw (54.center) to (55);
		\draw (56.center) to (57);
		\draw (60.center) to (61);
		\draw [in=180, out=90] (61) to (62.center);
		\draw [in=180, out=-90] (61) to (63.center);
		\draw (67.center) to (68);
		\draw [in=180, out=90] (68) to (69.center);
		\draw [in=180, out=-90] (68) to (70.center);
		\draw (71.center) to (72);
		\draw [in=180, out=90] (72) to (73.center);
		\draw [in=180, out=-90] (72) to (74.center);
	\end{pgfonlayer}
\end{tikzpicture} \]
\end{defiC}

A copy-delete category formalizes copying and discarding of information. In a hypergraph category, we have additional ways of combining morphisms 

\begin{defi}
A \emph{hypergraph category} (e.g. \cite{hypergraphcats}) is a copy-delete category where every object is additionally equipped with morphisms 
\[ \mult_X = 
	\tikzset{baseline=-0.65ex}
	\begin{tikzpicture}[scale=0.1]
	\begin{pgfonlayer}{nodelayer}
		\node [style=none] (0) at (0, 2) {};
		\node [style=bn] (1) at (2, 0) {};
		\node [style=none] (2) at (0, -2) {};
		\node [style=none] (3) at (5, 0) {};
		\node [style=none] (4) at (-2, 2) {};
		\node [style=none] (5) at (-2, -2) {};
	\end{pgfonlayer}
	\begin{pgfonlayer}{edgelayer}
		\draw (0.center) to (4.center);
		\draw [in=90, out=0] (0.center) to (1);
		\draw (5.center) to (2.center);
		\draw [in=-90, out=0] (2.center) to (1);
		\draw (1) to (3.center);
	\end{pgfonlayer}
\end{tikzpicture}}
 : X \otimes X \to X , \qquad \unit_X = 
	\tikzset{baseline=-0.65ex}
	\begin{tikzpicture}[scale=0.1]
	\begin{pgfonlayer}{nodelayer}
		\node [style=bn] (1) at (-1, 0) {};
		\node [style=none] (3) at (2, 0) {};
	\end{pgfonlayer}
	\begin{pgfonlayer}{edgelayer}
		\draw (1) to (3.center);
	\end{pgfonlayer}
\end{tikzpicture}
}
 : I \to X \]
which interact with copy and delete as a special commutative Frobenius algebra. We don't put the axioms here, but refer to \cite{hypergraphcats} for reference.
\end{defi}

In a hypergraph category, the directionality of morphisms is not of essential importance: Every hypergraph category is equivalent to its opposite. In a \emph{Markov category}, on the other hand, directionality matters, and information is propagated in one direction only. This is axiomatized by requiring that discarding the output of a computation is the same as not performing the computation at all.

\begin{defiC}[\cite{fritz}]
A \emph{Markov category} is a copy-delete category in which $\del$ is natural; that is, every morphism $f : X \to Y$ satisfies the discardability equation
\[ \begin{tikzpicture}[scale=0.2]
	\begin{pgfonlayer}{nodelayer}
		\node [style=bn] (1) at (-2, 0) {};
		\node [style=morphism] (3) at (-6, 0) {$f$};
		\node [style=none] (4) at (-10, 0) {};
		\node [style=none] (5) at (0, 0) {=};
		\node [style=bn] (6) at (6, 0) {};
		\node [style=none] (8) at (2, 0) {};
		\node [style=none] (9) at (4, 2.5) {$X$};
		\node [style=none] (10) at (-3.5, 2.5) {$Y$};
		\node [style=none] (11) at (-8.5, 2.5) {$X$};
	\end{pgfonlayer}
	\begin{pgfonlayer}{edgelayer}
		\draw (1) to (3);
		\draw (4.center) to (3);
		\draw (8.center) to (6);
	\end{pgfonlayer}
\end{tikzpicture} \]
\end{defiC}

In general, morphisms are not required to commute with copying. Those that do are called deterministic.

\begin{defiC}[{\cite[10.1]{fritz}}]\label{def:determinism}
A morphism $f : X \to Y$ in a Markov category $\C$ is called \emph{deterministic} if it satisfies the equation
\[ \begin{tikzpicture}[scale=0.2]
	\begin{pgfonlayer}{nodelayer}
		\node [style=bn] (1) at (-4, 0) {};
		\node [style=morphism] (3) at (-6.5, 0) {$f$};
		\node [style=none] (4) at (-9, 0) {};
		\node [style=none] (5) at (0, 0) {=};
		\node [style=none] (9) at (-2, 2) {};
		\node [style=none] (10) at (-2, -2) {};
		\node [style=bn] (11) at (5, 0) {};
		\node [style=none] (12) at (5, 0) {};
		\node [style=none] (13) at (2, 0) {};
		\node [style=morphism] (14) at (8, 2) {$f$};
		\node [style=morphism] (15) at (8, -2) {$f$};
		\node [style=none] (16) at (11, 2) {};
		\node [style=none] (17) at (11, -2) {};
	\end{pgfonlayer}
	\begin{pgfonlayer}{edgelayer}
		\draw (1) to (3);
		\draw (4.center) to (3);
		\draw [in=180, out=90] (1) to (9.center);
		\draw [in=-90, out=-180] (10.center) to (1);
		\draw (11) to (12.center);
		\draw (13.center) to (12.center);
		\draw [in=180, out=90] (11) to (14);
		\draw [in=-90, out=-180] (15) to (11);
		\draw (16.center) to (14);
		\draw (17.center) to (15);
	\end{pgfonlayer}
\end{tikzpicture} \]
\end{defiC}

\noindent Markov categories have been highly successful as an abstract theory of information flow. They are used to model a variety of uncertainty-related effects, notably probabilistic and nondeterministic computation \cite{fritz}, and fresh name generation \cite{fritz2023dilations}. We now recall relevant examples, before defining a new Markov category for lack of information and extended Gaussian computation.

\subsection{Relations}\label{sec:relations}

A basic example of a Markov category formalizes relational nondeterminism. We call a relation $R \subseteq X \times Y$ \emph{total} if for all $x \in X$ there exists a $y \in Y$ with $(x,y) \in R$. 

\begin{defi}
The Markov category $\catname{TRel}$ (called $\catname{SetMulti}$ in \cite{fritz})) has 
\begin{enumerate}
\item objects sets $X$
\item morphisms $X \to Y$ are total relations $R \subseteq X \times Y$
\item the identity relation is $\{ (x,x) \s x \in X \}$
\item composition is relation composition
\[ (R ; S) = \{ (x,z) \s \exists y, (x,y) \in R, (y,z) \in S \}\]
\item tensor is product $X \times Y$, with unit $I=1$ and
\[ (R_1 \otimes R_2) = \{ ((x_1,x_2),(y_1,y_2)) : (x_1,y_1) \in R_1, (x_2,y_2) \in R_2 \}\]
\item copy and delete are
\[ \cpy_X = \{ (x,(x,x)) \s x \in X \}, \quad \del_X = \{ (x,\bullet) \s x \in X \}\]
\end{enumerate}
\end{defi}
We note that verifying the defining property of Markov categories (the naturality of $\del_X$) makes use of the totality assumption on the relations. 

An important special case of nondeterminism is given by the notion of linear relation.

\begin{defi}
A \emph{linear relation} is a relation $\D \subseteq \R^m \times \R^n$ that is also a vector subspace. The category $\tlinrel$ has 
\begin{enumerate}
\item objects $\R^n$
\item morphisms $\R^m \to \R^n$ are total linear relations $\D \subseteq X \times Y$
\end{enumerate}
All other structure is inherited from $\catname{TRel}$.
\end{defi}

\subsection{Probability channels}\label{sec:channels}

We now turn measure-theoretic probability into a Markov category using the notion of a Markov kernel, which is in essence a parameterized probability measure.

\begin{defi}
Let $(X,\E), (Y,\F)$ be measurable spaces. A \emph{channel} or \emph{Markov kernel} $\kappa : (X,\E) \chanto (Y,\F)$ is a map $\kappa : X \times \F \to [0,1]$ such that
\begin{itemize}
\item for every $x \in X$, $\kappa(x,-) : \F \to [0,1]$ is a probability measure
\item for every $F \in \F$, $\kappa(-,\F) : (X,\E) \to ([0,1],\B([0,1]))$ is a measurable function
\end{itemize}
We write the value $\kappa(x,F)$ as $\kappa(F|x)$. 
\end{defi}

The following is a classical construction due to Lawvere (see e.g. \cite[Section~4]{fritz}).
\begin{defi}
Channels form a Markov category $\stoch$ as follows.
\begin{enumerate}
\item Objects are measurable spaces $(X,\E)$
\item morphisms are channels
\item the identity morphism is given by $\id(E|x) = [x \in E]$
\item composition is given by integration. If $\kappa : (X,\E) \chanto (Y,\F)$ and $\nu : (Y,\F) \chanto (Z,\mathcal H)$, then $\nu\kappa : (X,\E) \chanto (Z,\mathcal H)$ is defined by
\[ (\nu\kappa)(H|x) = \int \nu(H|y) \kappa(\mathrm dy|x) \]
\item the tensor of $(X,\E)$ and $(Y,\F)$ is $(X \times Y, \E \otimes \F)$ where $\E \otimes \F$ denotes the product $\sigma$-algebra, and
\[ (\kappa \otimes \kappa')(E \times E'|x, x') = \kappa(E|x) \cdot \kappa'(E'|x')\]
\item every measurable function $f : (X,\E) \to (Y,\F)$ gives rise to a channel $\lift f : (X,\E) \chanto (Y,\F)$ via $\lift{f} (F|x) = [f(x) \in F]$. This assignment defines a functor $\meas \to \stoch$, and the copy and delete maps are inherited from the corresponding measurable functions.
\end{enumerate}
\end{defi}

\subsection{Gaussian maps}\label{sec:gauss}

The Markov category $\gauss$ is a compact description of Gaussian probability. Its objects are the spaces $\R^n$, and morphisms are linear maps perturbed by Gaussian noise.

\begin{defiC}[{\cite[Section~6]{fritz}}]
The Markov category $\gauss$ is defined as follows
\begin{enumerate}
\item objects are the spaces $\R^n$
\item morphisms $m \to n$ are tuples $(M,\psi)$ where $M \in \R^{m \times n}$ is a matrix and $\psi$ is a Gaussian distribution on $\R^n$. We think of the tuple $(M,\psi)$ as representing the stochastic function $x \mapsto Mx + \psi$.
\item composition is defined by 
\[ (M,\psi) \circ (N,\chi) = (MN, \psi + M_*\chi)\]
where pushforward and convolution make use of the formulas of Proposition~\ref{prop:gauss_laws}.
\end{enumerate}
\end{defiC}

As expected, composition of Gaussian maps agrees with measure-theoretic composition of Gaussians. 

\begin{propC}[{\cite[Proposition~6.1]{fritz}}]
There is a faithful functor $J : \gauss \to \stoch$ which interprets $\R^n$ as the measurable space $(\R^n,\B(\R^n))$, and assigns the morphism $(M,\psi) \in \gauss(\R^m,\R^n)$ to the channel $\kappa : \R^m \times \B(\R^n) \to [0,1]$ defined by
\[ \kappa(E|x) = \psi(\{ y \s y+Mx \in E \}). \]
\end{propC}

\section{Extended Gaussian Maps}\label{sec:gaussex}

Gaussian maps could faithfully be interpreted as channels, by interpreting objects as the Borel spaces $(\R^n,\B(\R^n))$. If we want to work with extended Gaussian distributions, this will no longer be possible, because the $\sigma$-algebra on the objects depends on the distribution we are working with. 

We need to introduce a formalism which lets morphisms $X \to Y$ dynamically introduce a quotient on $Y$. We achieve this using a particular instance of the cospan construction, which we call \emph{copartial maps}. This will be a generalization of the kernel representation seen in Proposition~\ref{prop:kernel_rep}.

\subsection{Copartial maps}\label{sec:copartiality}

A partial function between sets $X \parto Y$ is a function from some subset $A \subseteq X$ to $Y$. We are interested in the dual concept; a \emph{copartial function} $X \coparto Y$ should be a function from $X$ into some quotient of $Y/{\sim}$. The way to formalize this is using the language of spans and cospans. \\

Let $\C$ be a category. A span is a pair of morphisms $X \xleftarrow{f} A \xrightarrow{g} Y$. We consider two such spans equivalent if there is an isomorphism $A \cong A'$ which commutes with the legs. We will always consider spans up to equivalence. The following subclasses of spans are of importance
\begin{enumerate}
\item a \emph{relation} is a span where $f,g$ are jointly monic, i.e. the pairing $\langle f,g \rangle : A \to X \times Y$ is a monomorphism 
\item a \emph{partial map} is a span where the left leg $f$ is monic. The apex $A$ represents the domain of the partial map.
\end{enumerate}
Spans compose by pullback. A pullback of relations is not necessarily a relation; instead an image factorization has to be used. For partial maps, the two types of composition agree -- a pullback of partial maps is another partial map. The category $\Par(\C)$ of \cite{cockett2002restriction} has the same objects as $\C$, and morphisms are partial maps. \\

Just as partial maps are maps out of subobjects, \emph{copartial maps} are maps into quotients. For this, we consider cospans $X \xrightarrow{f} P \xleftarrow{g} Y$ (again, always up to equivalence).

\begin{enumerate}
\item a \emph{co-relation} is a cospan where $f,g$ are jointly epic, i.e. the copairing $[f,g] : X + Y \to P$ is an epimorphism. Co-relations have been studied in \cite{fong:decorated_corelations}.
\item a \emph{copartial map} is a cospan whose right leg $g$ is epic. This represents a quotient on the codomain. 
\end{enumerate}
Cospans compose by pushout; again, co-relations are not closed under pushout, but copartial maps are. This is because epimorphisms are stable under pushout.

\begin{defi}
Let $\C$ be a category with pushouts. We define a category $\Copar(\C)$ as follows
\begin{enumerate}
\item objects are the same as $\C$
\item morphisms $X \coparto Y$ are copartial maps
\item the identity cospan is $X \xrightarrow\id X \xleftarrow\id X$.
\item composition is given by pushout
\[\begin{tikzcd}
	&& W \\
	& P && Q \\
	X && Y && Z
	\arrow[from=3-1, to=2-2]
	\arrow[two heads, from=3-3, to=2-2]
	\arrow[from=2-2, to=1-3]
	\arrow[two heads, from=2-4, to=1-3]
	\arrow[from=3-3, to=2-4]
	\arrow[two heads, from=3-5, to=2-4]
	\arrow["\lrcorner"{anchor=center, pos=0.125, rotate=-45}, draw=none, from=1-3, to=3-3]
\end{tikzcd}\]
\end{enumerate}
\end{defi}
Well-definedness is the same as for general categories of cospans \cite{hypergraphcats,fong:decorated_cospans}.  We remark that, as in \cite{cockett2002restriction}, the right leg may be chosen to lie in a suitable, distinguished class $\mathcal E$ of epimorphisms (e.g. regular epimorphisms in a regular category), but we won't need this in what follows. Some basic observations:

\begin{enumerate}
\item we have a functor $\C \to \Copar(\C)$ which sends $f : X \to Y$ to the cospan $X \xrightarrow{f} Y \xleftarrow{\id} Y$.
\item If $\C$ has an inital object $0$, then the category $\Copar(\C)$ has a monoidal structure whose tensor is the coproduct $+$, and the unit is the initial object $0$.
\item The states $0 \coparto X$ are in bijection with the quotients $X \twoheadrightarrow P$ of $X$.
\end{enumerate}

For example, in $\Copar(\catname{Set})$, we have a state $u_X : 0 \coparto X$ given by the cospan $0 \to 1 \leftarrow X$, which corresponds to forgetting all information about $X$ (everything is quotiented to a point). For any $F : X \coparto Y$ we have that $F \circ u_X = u_Y$, because the following is a pushout
\[\begin{tikzcd}
	&& 1 \\
	& 1 && Q \\
	0 && X && Y
	\arrow[from=3-1, to=2-2]
	\arrow[from=3-3, to=2-2]
	\arrow[from=3-3, to=2-4]
	\arrow[from=3-5, to=2-4]
	\arrow[from=2-2, to=1-3]
	\arrow[from=2-4, to=1-3]
	\arrow["\lrcorner"{anchor=center, pos=0.125, rotate=-45}, draw=none, from=1-3, to=3-3]
\end{tikzcd}\]
Intuitively, if we don't know anything about $X$, no amount of post-processing can change that. The $\Copar$ construction captures uncertainty via lack of information, but the way it does this is quite subtle.

\begin{enumerate}
\item How do copartial maps relate to relational nondeterminism? From a copartial map $F = (X \xrightarrow{f} P \xleftarrow{p} Y)$, we may extract the relation $R_{F} = \{ (x,y) \s f(x) = p(y) \}$ which captures which outcomes $y \in Y$ each $x \in X$ may be identified with. Is it the case that $R$ respects composition, i.e. defines a functor $R : \Copar(\set) \to \rel$? The answer is \textbf{negative}; relational nondeterminism and copartiality compose differently. We spell out a counterexample in the appendix (Section~\ref{app:copartiality}). 

\item Can we turn $\Copar$ into a Markov category? Here, the choice of monoidal structure $(+,0)$ gives us no obvious copy and delete morphisms. It is tempting to try and define a monoidal structure with respect to $(\times, 1)$ but these don't generally respect pushout. Again, a counterexample for $\Copar(\set)$ is spelled out in the appendix. 
\end{enumerate}
The second issue disappears if the category $\C$ has biproducts. In this case, $(\times,1) = (+,0)$, and $(\Copar(\C),\times,1)$ does become a Markov category, with copy and delete being
\[ X \xrightarrow{\Delta} X \times X \xleftarrow{\id} X \times X, \qquad X \xrightarrow{!} 1 \xleftarrow{\id} 1. \]

\subsection{Linear Relations as Cospans}\label{sec:linrel}

Let $\vect$ denote the category of finite-dimensional vector spaces $\R^n$; this category has biproducts. We claim that the Markov category $\Copar(\vect)$ is equivalent to the category $\tlinrel$ of total linear relations. Unlike in $\set$, quotient-based and relational nondeterminism happen to coincide. 

\begin{prop}
A kernel representation of a linear relation $\fib L \subseteq X \times Y$ is a cospan of linear maps $X \xrightarrow{f} P \xleftarrow{p} Y$ such that $\fib L = \{ (x,y) \s f(x) = p(y) \}$. Every linear relation admits a kernel representation. 
\end{prop}

\noindent Kernel representations are not unique as cospans. They become unique if we restrict the copairing $[f,g] : X \times Y \to P$ to be jointly surjective, i.e. a co-relation. The relationship between linear relations and cospans is well-known (e.g. \cite{fong:decorated_cospans,fong:decorated_corelations,fong:open_systems}). In the case of copartial maps, the result specializes as follows: 

\begin{prop}
A linear relation $\fib L \subseteq X \times Y$ is total if and only if it admits a kernel representation $X \xrightarrow{f} P \xleftarrow{p} Y$ that is a copartial map. 
\end{prop}
\begin{proof}
Necessity is clear; if $p$ is surjective, then for every $x$ there exists $y$ with $p(y) = f(x)$, hence $(x,y) \in \fib L$. Conversely, if $\fib L$ is total, consider the cospan $X \xrightarrow{f} (X \times Y)/\fib L \xleftarrow{p} Y$ given by $f(x) = (x,0) + \fib L$, $p(y) = (0,y) + \fib L$. We claim that $p$ is surjective: For any pair $(x,y)$, we can using totality find $y'$ with $(x,y') \in \fib L$. Then $[p(y-y')] = [(x,y) - (x,y')] = [(x,y)]$ as desired. 
\end{proof}

\noindent In particular, we have an isomorphism of cospans
\[\begin{tikzcd}
	& {(X\times Y)/\fib L} \\
	{X} && {Y} \\
	& {Y/\D}
	\arrow["p"', two heads, from=2-3, to=1-2]
	\arrow["\ell"', from=2-1, to=3-2]
	\arrow["q", two heads, from=2-3, to=3-2]
	\arrow["\sim"{description}, from=1-2, to=3-2]
	\arrow["f", from=2-1, to=1-2]
\end{tikzcd}\]
where $\D = \ker(p)$. The space $\D$ plays the role of the fibre in extended Gaussian distributions. This gives us the following useful representation: Write $\fib L(x) = \{ y \s (x,y) \in \fib L \}$, then 
\[ \fib L(x) = \ell(x) + \D \label{eq:linrel_nf} \]
That is to say, we can see a every total linear relation as a linear map plus some nondeterministic noise. This representation is very useful in the categorical setting:

\begin{exa}
Every subspace $\D \subseteq Y$ induces a state $\D : 0 \to Y$ in $\tlinrel$. Every linear map $\ell : X \to Y$ induces a linear relation via its graph, $\lift \ell = \{ (x,y) \s y = \ell(x) \}$. Then every total linear relation $\fib L : X \to Y$ can be written as a composite
\begin{equation} \begin{tikzpicture}[scale=0.2]
	\begin{pgfonlayer}{nodelayer}
		\node [style=morphism] (6) at (-12, 0) {$\fib L$};
		\node [style=none] (7) at (-16, 0) {};
		\node [style=none] (8) at (-8, 0) {};
		\node [style=none] (9) at (-6, 0) {=};
		\node [style=none] (10) at (-4, -2) {};
		\node [style=morphism] (11) at (1, -2) {$\lift{\ell}$};
		\node [style=none] (12) at (6, -2) {};
		\node [style=wn] (13) at (8, 0) {};
		\node [style=none] (14) at (6, 2) {};
		\node [style=morphism] (15) at (1, 2) {$\fib D$};
		\node [style=none] (16) at (13, 0) {};
	\end{pgfonlayer}
	\begin{pgfonlayer}{edgelayer}
		\draw (7.center) to (6);
		\draw (6) to (8.center);
		\draw (10.center) to (11);
		\draw (11) to (12.center);
		\draw (15) to (14.center);
		\draw [in=135, out=0, looseness=0.75] (14.center) to (13);
		\draw [in=225, out=0] (12.center) to (13);
		\draw (16.center) to (13);
	\end{pgfonlayer}
\end{tikzpicture} \label{eq:linrel_add} \end{equation}
where the white dot denotes the addition relation $\{ (x_1,x_2,y) \s y=x_1+x_2\}$. Note that addition (together with zero) defines a monoid structure on each object of $\tlinrel$. 
\end{exa}

\begin{prop}\label{prop:trel_copar}
The Markov category $\tlinrel$ of total linear relations is equivalent to $\Copar(\vect)$.
\end{prop}
\begin{proof}
This is discussed in \cite[Section~7.2]{fong:decorated_corelations} for arbitrary linear relations. In our case, if we have $\fib L(x) = f(x) + \D$ and $\fib M(x) = g(x) + \fib E$, then their composite is
\[ (\fib L ; \fib M)(x) = g(\ell(x)) + g[\D] + \fib E \]
The same happens in the pushout of kernel representations
\[\begin{tikzcd}
	&& W \\
	& {Y/\D} && {Z/\fib E} \\
	X && Y && Z
	\arrow["f", from=3-1, to=2-2]
	\arrow[two heads, from=3-3, to=2-2]
	\arrow["g", from=3-3, to=2-4]
	\arrow[two heads, from=3-5, to=2-4]
	\arrow[from=2-2, to=1-3]
	\arrow[two heads, from=2-4, to=1-3]
\end{tikzcd}\]
which is given by $W \cong Z/(g[\D] + \fib E)$. 
\end{proof}

\subsection{Extended Gaussian maps}\label{sec:gaussex_def}

To define extended Gaussian maps as formal objects, we combine the cospan approach to linear relations with Gaussian distributions. For this, we employ the concept of a \emph{decorated cospan} \cite{fong:decorated_cospans,fong:decorated_corelations}.

\begin{defi}
An extended Gaussian morphism $\R^m \to \R^n$ is an equivalence class of tuples $(f,\psi,q)$ where $\R^m \xrightarrow{f} \R^k \xleftarrow{p} \R^n$ is a copartial map, and $\psi \in \gauss(\R^k)$ is a Gaussian distribution. We call $\psi$ the \emph{decoration} of the cospan. We identify a decorated cospan $(f,\psi,q)$ with $(f,\psi',q)$ when there exists a linear isomorphism of cospans $\tau$ with $\tau_*\psi = \psi'$.
\end{defi}

\noindent Intuitively, we think of $(f,\psi,q)$ as encoding a linear map with extended Gaussian noise
\[ x \mapsto f(x) + q^{-1}(\psi) \]

\begin{defi}
The Markov category $\gaussex$ is defined as follows
\begin{enumerate}
\item objects are the spaces $\R^n$
\item morphisms are decorated copartial maps $(f,\psi,q)$
\item the identity is $(\id, \delta_0, \id)$, 
\item the composite of cospans $(f,\psi,p)$ and $(g,\chi,q)$ is $(i_1f,\xi,i_2q)$ where we form the pushout
\[\begin{tikzcd}
	&& W \\
	& P && Q \\
	X && Y && Z
	\arrow["q", two heads, from=3-5, to=2-4]
	\arrow["p", two heads, from=3-3, to=2-2]
	\arrow["{i_1}", from=2-2, to=1-3]
	\arrow["g"', from=3-3, to=2-4]
	\arrow["{i_2}"', two heads, from=2-4, to=1-3]
	\arrow["f", from=3-1, to=2-2]
	\arrow["\lrcorner"{anchor=center, pos=0.125, rotate=-45}, draw=none, from=1-3, to=3-3]
\end{tikzcd}\]
The composite decoration $\xi$ is the convolution of pushforwards $(i_1)_*\psi + (i_2)_*\chi$
\item Every linear map $f : \R^m \to \R^n$ induces an extended Gaussian morphism $(f,\N(0,0),\id)$. Copy and delete structure are inherited from vector spaces that way.
\end{enumerate}
\end{defi}

\noindent It is easy to see that the category $\gaussex$ faithfully contains Gaussian probability as well as relational nondeterminism.

\begin{prop}\label{prop:embedding}
We have faithful Markov functors from $\gauss$ and $\tlinrel$ to $\gaussex$:
\begin{enumerate}
\item A Gaussian map $(f,\psi)$ gets sent to the decorated copartial map $(f,\psi,\id)$.
\item A total linear relation represented as a copartial map $X \xrightarrow{\ell} P \xleftarrow{p} Y$ gets sent ot the decorated copartial map $(f,\delta_0,p)$.
\end{enumerate}
\end{prop}

\noindent To give a morphism in $\gaussex(0,\R^n)$ is to give an extended Gaussian distribution on $\R^n$. Furthermore, every $\gaussex$ morphism can be written as a sum of a linear map and an extended Gaussian distribution in the style of \eqref{eq:linrel_add}.

\begin{prop}\label{prop:gaussex_state_rules}
The transformation rules of Definition~\ref{def:gaussex_state_rules} are valid in $\gaussex$. 
\end{prop}
\begin{proof}
Using the embeddings of Proposition~\ref{prop:embedding}, this can be reduced to the way subspaces and Gaussian distributions compose in the subcategories $\gauss$ and $\tlinrel$, individually. 
\end{proof}

\subsection{Example: String Diagrams}

We give an example of using the categorical structure of $\gaussex$ to create distributions from simple building blocks. The language of string diagrams lets us visualize this interconnection in an intuitive way.

\begin{exa}\label{ex:stringdiag}
We can describe the joint distribution $P_{VI} \in \gaussex(\R^2)$ of the noisy resistor through the following process: Initialize $I$ to take any value in $\R$, then let $V = RI + \epsilon$ where $\epsilon \sim \N(0,\sigma^2)$. In statistical notation, we would define this joint distribution as
\begin{align*}
I &\sim \R \\
V &\sim \N(RI,\sigma^2)
\end{align*}
Using string diagrams, we can describe the distribution $P_{VI}$ as the composite 
\[ \begin{tikzpicture}[scale=0.2]
	\begin{pgfonlayer}{nodelayer}
		\node [style=bn] (0) at (-10, -21) {};
		\node [style=bn] (1) at (-5, -21) {};
		\node [style=none] (2) at (-3, -23) {};
		\node [style=none] (3) at (-3, -19) {};
		\node [style=morphism] (4) at (2, -19) {$\times R$};
		\node [style=none] (5) at (8, -19) {};
		\node [style=wn] (6) at (10, -17) {};
		\node [style=none] (7) at (8, -15) {};
		\node [style=morphism] (8) at (2, -15) {$\N(0,\sigma^2)$};
		\node [style=none] (9) at (13, -17) {};
		\node [style=none] (10) at (13, -23) {};
		\node [style=none] (11) at (-19, -19) {$P_{VI}$};
		\node [style=none] (12) at (-15, -19) {=};
		\node [style=none] (13) at (-8, -24) {$I$};
		\node [style=none] (14) at (15, -23) {$I$};
		\node [style=none] (15) at (15, -17) {$V$};
	\end{pgfonlayer}
	\begin{pgfonlayer}{edgelayer}
		\draw (0) to (1);
		\draw [in=-180, out=-45] (1) to (2.center);
		\draw (2.center) to (10.center);
		\draw [in=-180, out=45] (1) to (3.center);
		\draw (3.center) to (4);
		\draw (4) to (5.center);
		\draw [in=-135, out=0, looseness=0.75] (5.center) to (6);
		\draw [in=0, out=135] (6) to (7.center);
		\draw (7.center) to (8);
		\draw (6) to (9.center);
	\end{pgfonlayer}
\end{tikzpicture} \]
where $
	\tikzset{baseline=-0.65ex}
	\begin{tikzpicture}[scale=0.1]
	\begin{pgfonlayer}{nodelayer}
		\node [style=wn] (1) at (0, 0) {};
		\node [style=none] (3) at (3, 0) {};
		\node [style=none] (4) at (-3, 1) {};
		\node [style=none] (5) at (-3, -1) {};
	\end{pgfonlayer}
	\begin{pgfonlayer}{edgelayer}
		\draw (1) to (3.center);
		\draw [in=120, out=0] (4.center) to (1);
		\draw [in=-120, out=0] (5.center) to (1);
	\end{pgfonlayer}
\end{tikzpicture}}
$ denotes addition as before, and $
	\tikzset{baseline=-0.65ex}
	\begin{tikzpicture}[scale=0.1]
	\begin{pgfonlayer}{nodelayer}
		\node [style=bn] (1) at (-1, 0) {};
		\node [style=none] (3) at (2, 0) {};
	\end{pgfonlayer}
	\begin{pgfonlayer}{edgelayer}
		\draw (1) to (3.center);
	\end{pgfonlayer}
\end{tikzpicture}
}
$ denotes initializing a completely unspecified variable (i.e. the relation $\R$). We have annotated the wires with labels $V,I$ corresponding to the output variables for easier reading, but these labels have no formal meaning.

We can now manipulate these string diagrams to compute other quantities of interest. For example, to find the marginal distribution $P_I$, we discard the $V$-wire, i.e. compose with the discard map $
	\tikzset{baseline=-0.65ex}
	\begin{tikzpicture}[scale=0.1]
	\begin{pgfonlayer}{nodelayer}
		\node [style=bn] (1) at (2, 0) {};
		\node [style=none] (3) at (-1, 0) {};
	\end{pgfonlayer}
	\begin{pgfonlayer}{edgelayer}
		\draw (1) to (3.center);
	\end{pgfonlayer}
\end{tikzpicture}
}
$, and simplify using the laws of Markov categories. 
\[ \begin{tikzpicture}[scale=0.2]
	\begin{pgfonlayer}{nodelayer}
		\node [style=bn] (0) at (-10, -21) {};
		\node [style=bn] (1) at (-5, -21) {};
		\node [style=none] (2) at (-3, -23) {};
		\node [style=none] (3) at (-3, -19) {};
		\node [style=morphism] (4) at (2, -19) {$\times R$};
		\node [style=none] (5) at (8, -19) {};
		\node [style=wn] (6) at (10, -17) {};
		\node [style=none] (7) at (8, -15) {};
		\node [style=morphism] (8) at (2, -15) {$\N(0,\sigma^2)$};
		\node [style=bn] (9) at (13, -17) {};
		\node [style=none] (10) at (18, -23) {};
		\node [style=none] (11) at (-17.5, -19) {$P_{I}$};
		\node [style=none] (12) at (-14, -19) {=};
		\node [style=none] (24) at (19, -19) {=};
		\node [style=bn] (25) at (51.5, -19) {};
		\node [style=none] (26) at (56.5, -19) {};
		\node [style=none] (28) at (47, -19) {=};
		\node [style=bn] (29) at (23, -21) {};
		\node [style=bn] (30) at (28, -21) {};
		\node [style=none] (31) at (30, -23) {};
		\node [style=none] (32) at (30, -19) {};
		\node [style=morphism] (33) at (35, -19) {$\times R$};
		\node [style=bn] (34) at (42, -19) {};
		\node [style=bn] (36) at (42, -15) {};
		\node [style=morphism] (37) at (35, -15) {$\N(0,\sigma^2)$};
		\node [style=none] (39) at (45, -23) {};
	\end{pgfonlayer}
	\begin{pgfonlayer}{edgelayer}
		\draw (0) to (1);
		\draw [in=-180, out=-45] (1) to (2.center);
		\draw (2.center) to (10.center);
		\draw [in=-180, out=45] (1) to (3.center);
		\draw (3.center) to (4);
		\draw (4) to (5.center);
		\draw [in=-135, out=0, looseness=0.75] (5.center) to (6);
		\draw [in=0, out=135] (6) to (7.center);
		\draw (7.center) to (8);
		\draw (6) to (9);
		\draw (25) to (26.center);
		\draw (29) to (30);
		\draw [in=-180, out=-45] (30) to (31.center);
		\draw (31.center) to (39.center);
		\draw [in=-180, out=45] (30) to (32.center);
		\draw (32.center) to (33);
		\draw (33) to (34);
		\draw (36) to (37);
	\end{pgfonlayer}
\end{tikzpicture} \]
We have recovered the fact that the variable $I$ has distribution $\R$. In order to compute the pushforward $f_*P_{IV}$ under the map $f(V,I) = V-RI$, we reason
\[ \hspace*{-4pt} \begin{tikzpicture}[scale=0.2]
	\begin{pgfonlayer}{nodelayer}
		\node [style=bn] (0) at (-10, -21) {};
		\node [style=bn] (1) at (-5, -21) {};
		\node [style=none] (2) at (-3, -23) {};
		\node [style=none] (3) at (-3, -19) {};
		\node [style=morphism] (4) at (2, -19) {$\times R$};
		\node [style=none] (5) at (8, -19) {};
		\node [style=wn] (6) at (10, -17) {};
		\node [style=none] (7) at (8, -15) {};
		\node [style=morphism] (8) at (-7.5, -15) {$\N(0,\sigma^2)$};
		\node [style=none] (9) at (25, -17) {};
		\node [style=morphism] (10) at (22, -23) {$\times (-R)$};
		\node [style=none] (11) at (-20, -19) {$(P_{VI} ; f)$};
		\node [style=none] (12) at (-15, -19) {=};
		\node [style=wn] (13) at (28.75, -20) {};
		\node [style=none] (14) at (35, -20) {};
		\node [style=none] (15) at (25, -23) {};
		\node [style=none] (39) at (12, -26) {};
		\node [style=none] (40) at (12, -12) {};
		\node [style=none] (41) at (-13, -12) {};
		\node [style=none] (42) at (-13, -26) {};
		\node [style=none] (43) at (32, -26) {};
		\node [style=none] (44) at (32, -12) {};
		\node [style=none] (45) at (15, -12) {};
		\node [style=none] (46) at (15, -26) {};
		\node [style=none] (51) at (53, -20) {};
		\node [style=morphism] (52) at (43.5, -20) {$\N(0,\sigma^2)$};
		\node [style=none] (54) at (37, -20) {=};
	\end{pgfonlayer}
	\begin{pgfonlayer}{edgelayer}
		\draw (0) to (1);
		\draw [in=-180, out=-45] (1) to (2.center);
		\draw (2.center) to (10);
		\draw [in=-180, out=45] (1) to (3.center);
		\draw (3.center) to (4);
		\draw (4) to (5.center);
		\draw [in=-135, out=0, looseness=0.75] (5.center) to (6);
		\draw [in=0, out=135] (6) to (7.center);
		\draw (7.center) to (8);
		\draw (6) to (9.center);
		\draw [in=-225, out=0] (9.center) to (13);
		\draw (14.center) to (13);
		\draw [in=-135, out=0] (15.center) to (13);
		\draw (15.center) to (10);
		\draw [style=dashed box] (42.center) to (41.center);
		\draw [style=dashed box] (41.center) to (40.center);
		\draw [style=dashed box] (40.center) to (39.center);
		\draw [style=dashed box] (39.center) to (42.center);
		\draw [style=dashed box] (46.center) to (45.center);
		\draw [style=dashed box] (45.center) to (44.center);
		\draw [style=dashed box] (44.center) to (43.center);
		\draw [style=dashed box] (43.center) to (46.center);
		\draw (51.center) to (52);
	\end{pgfonlayer}
\end{tikzpicture} \]
by simplifying the following composite of linear functions
\[ \begin{tikzpicture}[scale=0.2]
	\begin{pgfonlayer}{nodelayer}
		\node [style=none] (0) at (-26, -1) {};
		\node [style=bn] (1) at (-23, -1) {};
		\node [style=none] (2) at (-21, -3) {};
		\node [style=none] (3) at (-21, 1) {};
		\node [style=morphism] (4) at (-19, 1) {$\times R$};
		\node [style=none] (5) at (-16, 1) {};
		\node [style=wn] (6) at (-14, 3) {};
		\node [style=none] (7) at (-16, 5) {};
		\node [style=none] (8) at (-26, 5) {};
		\node [style=none] (9) at (-9, 3) {};
		\node [style=morphism] (10) at (-12, -3) {$\times (-R)$};
		\node [style=wn] (13) at (-5.25, 0) {};
		\node [style=none] (14) at (-2, 0) {};
		\node [style=none] (15) at (-9, -3) {};
		\node [style=none] (28) at (40, 0) {=};
		\node [style=none] (29) at (43, -1) {};
		\node [style=bn] (30) at (46, -1) {};
		\node [style=none] (35) at (49, 2) {};
		\node [style=none] (36) at (43, 2) {};
		\node [style=none] (37) at (2, -2) {};
		\node [style=bn] (38) at (5, -2) {};
		\node [style=none] (39) at (7, -4) {};
		\node [style=none] (40) at (7, 0) {};
		\node [style=morphism] (41) at (10, 0) {$\times R$};
		\node [style=none] (42) at (13.5, -4) {};
		\node [style=wn] (43) at (15.5, -2) {};
		\node [style=none] (44) at (13.5, 0) {};
		\node [style=none] (45) at (2, 4) {};
		\node [style=none] (46) at (17.5, 4) {};
		\node [style=morphism] (47) at (10, -4) {$\times (-R)$};
		\node [style=wn] (48) at (19.25, 1) {};
		\node [style=none] (49) at (22.5, 1) {};
		\node [style=none] (50) at (17.5, -2) {};
		\node [style=none] (51) at (0, 0) {=};
		\node [style=none] (52) at (26, -2) {};
		\node [style=bn] (53) at (29, -2) {};
		\node [style=wn] (58) at (31, -2) {};
		\node [style=none] (60) at (26, 4) {};
		\node [style=none] (61) at (33, 4) {};
		\node [style=wn] (63) at (34.75, 1) {};
		\node [style=none] (64) at (38, 1) {};
		\node [style=none] (65) at (33, -2) {};
		\node [style=none] (66) at (24, 0) {=};
	\end{pgfonlayer}
	\begin{pgfonlayer}{edgelayer}
		\draw (0.center) to (1);
		\draw [in=-180, out=-45] (1) to (2.center);
		\draw (2.center) to (10);
		\draw [in=-180, out=45] (1) to (3.center);
		\draw (3.center) to (4);
		\draw (4) to (5.center);
		\draw [in=-135, out=0, looseness=0.75] (5.center) to (6);
		\draw [in=0, out=135] (6) to (7.center);
		\draw (7.center) to (8.center);
		\draw (6) to (9.center);
		\draw [in=-225, out=0] (9.center) to (13);
		\draw (14.center) to (13);
		\draw [in=-135, out=0] (15.center) to (13);
		\draw (15.center) to (10);
		\draw (29.center) to (30);
		\draw (35.center) to (36.center);
		\draw (37.center) to (38);
		\draw [in=-180, out=-45] (38) to (39.center);
		\draw (39.center) to (47);
		\draw [in=-180, out=45] (38) to (40.center);
		\draw (40.center) to (41);
		\draw [in=-135, out=0, looseness=0.75] (42.center) to (43);
		\draw [in=0, out=135] (43) to (44.center);
		\draw [in=-225, out=0] (46.center) to (48);
		\draw (49.center) to (48);
		\draw [in=-135, out=0] (50.center) to (48);
		\draw (41) to (44.center);
		\draw (42.center) to (47);
		\draw (45.center) to (46.center);
		\draw (43) to (50.center);
		\draw (52.center) to (53);
		\draw [in=-225, out=0] (61.center) to (63);
		\draw (64.center) to (63);
		\draw [in=-135, out=0] (65.center) to (63);
		\draw (60.center) to (61.center);
		\draw (58) to (65.center);
	\end{pgfonlayer}
\end{tikzpicture} \]
The equations used here are associativity of addition $
	\tikzset{baseline=-0.65ex}
	\begin{tikzpicture}[scale=0.1]
	\begin{pgfonlayer}{nodelayer}
		\node [style=wn] (1) at (0, 0) {};
		\node [style=none] (3) at (3, 0) {};
		\node [style=none] (4) at (-3, 1) {};
		\node [style=none] (5) at (-3, -1) {};
	\end{pgfonlayer}
	\begin{pgfonlayer}{edgelayer}
		\draw (1) to (3.center);
		\draw [in=120, out=0] (4.center) to (1);
		\draw [in=-120, out=0] (5.center) to (1);
	\end{pgfonlayer}
\end{tikzpicture}}
$, $rx + (-r)x = 0$, and unitality of zero. This style of reasoning is known as graphical linear algebra, and has been axiomatized soundly and completely \cite{graphical_la}. We return to the question of axiomatizing extended Gaussian distributions in Section~\ref{sec:presentations}.
\end{exa}

\section{Interconnection}\label{sec:interconnection}

After discussing extended Gaussian maps as abstract categorical entities, we will now tie things back to Willems's measure-theoretic approach to composition. 

We have seen that convolution (addition) of extended Gaussian distributions adds their fibres, i.e. increases uncertainty. There is a dual operation to this called \emph{interconnection} by Willems, which decreases uncertainty, i.e. refines the $\sigma$-algebra. We demonstrate this with an example: Consider again the noisy resistor
\[ V-RI = \epsilon \]
with $\epsilon \sim \N(0,\sigma^2)$. If we externally supply a fixed voltage $V_0$, we can solve for $I$ and obtain
\begin{equation}
	I = \frac 1 R (V_0-\epsilon) \label{eq:interconnection}
\end{equation} 

\noindent After this interconnection, $I$ becomes a classical (Gaussian) random variable with distribution $\N(\frac{V_0}{R},\frac{\sigma^2}{R^2})$. Interconnecting the systems has provided \emph{more} information and hence refined the $\sigma$-algebra associated to $I$. The string-diagrammatic way of representing interconnection is using a connective $
	\tikzset{baseline=-0.65ex}
	\begin{tikzpicture}[scale=0.1]
	\begin{pgfonlayer}{nodelayer}
		\node [style=none] (0) at (0, 2) {};
		\node [style=bn] (1) at (2, 0) {};
		\node [style=none] (2) at (0, -2) {};
		\node [style=none] (3) at (5, 0) {};
		\node [style=none] (4) at (-2, 2) {};
		\node [style=none] (5) at (-2, -2) {};
	\end{pgfonlayer}
	\begin{pgfonlayer}{edgelayer}
		\draw (0.center) to (4.center);
		\draw [in=90, out=0] (0.center) to (1);
		\draw (5.center) to (2.center);
		\draw [in=-90, out=0] (2.center) to (1);
		\draw (1) to (3.center);
	\end{pgfonlayer}
\end{tikzpicture}}
$. 

\[ \begin{tikzpicture}[scale=0.2]
	\begin{pgfonlayer}{nodelayer}
		\node [style=state] (0) at (-6, 0) {$P_{VI}$};
		\node [style=none] (1) at (-5, 2) {};
		\node [style=none] (2) at (-5, -2) {};
		\node [style=none] (3) at (1, 2) {};
		\node [style=none] (4) at (8, -2) {};
		\node [style=bn] (5) at (4, 5) {};
		\node [style=none] (6) at (1, 8) {};
		\node [style=none] (7) at (8, 5) {};
		\node [style=state] (8) at (-2, 8) {$V_0$};
	\end{pgfonlayer}
	\begin{pgfonlayer}{edgelayer}
		\draw (3.center) to (1.center);
		\draw (2.center) to (4.center);
		\draw (7.center) to (5);
		\draw [in=-90, out=0, looseness=1.25] (3.center) to (5);
		\draw [in=0, out=90, looseness=1.25] (5) to (6.center);
		\draw (8) to (6.center);
	\end{pgfonlayer}
\end{tikzpicture} \]

\noindent In statistical terms, this connective should represent conditioning on a random variable. But conditioning in categorical setting is notoriously a difficult subject, which has for example been explored in \cite{2021compositional,dilavore2023evidential}. The morphism $
	\tikzset{baseline=-0.65ex}
	\begin{tikzpicture}[scale=0.1]
	\begin{pgfonlayer}{nodelayer}
		\node [style=none] (0) at (0, 2) {};
		\node [style=bn] (1) at (2, 0) {};
		\node [style=none] (2) at (0, -2) {};
		\node [style=none] (3) at (5, 0) {};
		\node [style=none] (4) at (-2, 2) {};
		\node [style=none] (5) at (-2, -2) {};
	\end{pgfonlayer}
	\begin{pgfonlayer}{edgelayer}
		\draw (0.center) to (4.center);
		\draw [in=90, out=0] (0.center) to (1);
		\draw (5.center) to (2.center);
		\draw [in=-90, out=0] (2.center) to (1);
		\draw (1) to (3.center);
	\end{pgfonlayer}
\end{tikzpicture}}
$ is part of the structure of hypergraph categories, but is generally not supported in Markov categories like $\gaussex$. 

In \cite{willems:oss}, Willems discusses a special case of interconnection of Gaussian systems where the information provided by the component systems is in some sense \emph{complementary}, and can thus be interconnected without requiring conditioning. We recall this notion of complementary interconnection, and show that it agrees with categorical composition of extended Gaussian maps. We will then return to the issue of general, non-complementary interconnections via conditioning in Section~\ref{sec:conditioning}.

\subsection{Complementarity}\label{sec:complementarity}

\begin{defiC}[{\cite{willems:oss}}]
Two probabilistic systems $(X,\E_1,P_1), (X,\E_2,P_2)$ on the same underlying set are called \emph{complementary} if for all $E_1,E_1' \in \E_1$ and $E_2, E_2' \in \E_2$ we have
\[ E_1 \cap E_2 = E_1' \cap E_2' \Rightarrow E_1 = E_1' \text{ and } E_2 = E_2' \]
In this case, their interconnection is the system $(X,\E,P)$ where $\E$ is the $\sigma$-algebra generated by the intersections $E_1 \cap E_2$ for $E_1 \in \E_1, E_2 \in \E_2$, and we define a probability measure by
\[ P(E_1 \cap E_2) = P_1(E_1)P_2(E_2)\]
This is well-defined because of the complementarity condition.
\end{defiC}

\begin{exa}\label{ex:weakening}
A typical instance is the construction of the product distribution. Given systems $(X,\E,P_X)$ and $(Y,\F,P_Y)$, we can weaken them to use the same underlying set $X \times Y$; namely we define
\[ \E_1 = \{ E \times Y \s E \in \E \}, \E_2 = \{ X \times F \s F \in \F \}\]
and $P_1(E \times Y) = P_X(E)$, $P_2(X \times F) = P_2(F)$. Then the $\sigma$-algebras $\E_1$ and $\E_2$ are complementary, and the intersections are precisely the measurable rectangles $E \times F$ with $E \in \E, F \in \F$. 
\end{exa}

\noindent We can thus see interconnection is a generalization of the product distribution, which combines the information provided by the two distributions $P_1,P_2$ in an independent fashion.

\begin{propC}[{\cite[Section~VII]{willems:oss}}]\label{prop:linear_complementary}
Two linear systems on $\R^n$ with fibres $\D_1$ and $\D_2$ are complementary if and only if $\D_1 + \D_2 = \R^n$. In this case, the interconnected system has fibre $\D_1 \cap \D_2$. Interconnections of Gaussian systems are Gaussian. 
\end{propC}

\begin{exa}
In the noisy resistor example, we are interconnecting a system on $\R^2$ with fibre $\D = \{ (V,I) : V=RI\}$ with the deterministic system corresponding to $V=V_0$; its fibre is $\fib E = 0 \times \R$, because no information about $I$ is given. For $R \neq 0$, we have that $\D + \fib E = \R^2$, that is the systems are indeed complementary. The fibre of the interconnected system is $\D \cap \fib E = 0$, hence the system has become classical after interconnection (see \eqref{eq:interconnection}).
\end{exa}

\noindent Proposition~\ref{prop:linear_complementary} explains Willems' terminology \emph{open} and \emph{closed}: Open systems lack some information, which we can fill in via interconnection, thereby decreasing the fibre. Closed systems, whose fibre already vanishes, are complementary only to trivial systems, and hence cannot be refined by interconnecting with anything anymore. 

\subsection{Composition as Interconnection}\label{sec:composition}

We now relate Willems' composition by interconnection to categorical composition in $\gaussex$. Gaussian systems are a formalism of states (distributions), while $\gaussex$ features morphisms $f: X \to Y$ with dedicated inputs and outputs. The following well-known trick (e.g. \cite{abramsky2009categorical}) lets us turn a morphisms of a suitable categories into a state:

\begin{defi}
The \emph{name} $\lceil f \rceil : I \to X \otimes Y$ of a $\gaussex$ morphism $f : X \to Y$ is given by the following composite
\[ \begin{tikzpicture}[scale=0.4]
	\begin{pgfonlayer}{nodelayer}
		\node [style=state] (0) at (-10, 0) {$\lceil f \rceil$};
		\node [style=none] (3) at (-3.5, 0) {=};
		\node [style=morphism] (4) at (4, -1) {$f$};
		\node [style=none] (5) at (2, -1) {};
		\node [style=none] (6) at (6, -1) {};
		\node [style=bn] (7) at (0, 0) {};
		\node [style=bn] (8) at (-1, 0) {};
		\node [style=none] (9) at (2, 1) {};
		\node [style=none] (10) at (6, 1) {};
		\node [style=none] (11) at (-9, 0.75) {};
		\node [style=none] (12) at (-9, -0.75) {};
		\node [style=none] (13) at (-6, 0.75) {};
		\node [style=none] (14) at (-6, -0.75) {};
	\end{pgfonlayer}
	\begin{pgfonlayer}{edgelayer}
		\draw (6.center) to (4);
		\draw (4) to (5.center);
		\draw (10.center) to (9.center);
		\draw [in=90, out=-180] (9.center) to (7);
		\draw (7) to (8);
		\draw [in=-180, out=-90] (7) to (5.center);
		\draw (11.center) to (13.center);
		\draw (14.center) to (12.center);
	\end{pgfonlayer}
\end{tikzpicture} \]
where the unit $
	\tikzset{baseline=-0.65ex}
	\begin{tikzpicture}[scale=0.1]
	\begin{pgfonlayer}{nodelayer}
		\node [style=bn] (1) at (-1, 0) {};
		\node [style=none] (3) at (2, 0) {};
	\end{pgfonlayer}
	\begin{pgfonlayer}{edgelayer}
		\draw (1) to (3.center);
	\end{pgfonlayer}
\end{tikzpicture}
}
$ denotes uninformative distribution $u_X$ on $X$. We can interpret this state as a Gaussian system: Concretely, for an extended Gaussian morphism $(f,\psi,p) \in \gaussex(\R^m,\R^n)$ its name is the Gaussian system $(\R^m \times \R^n, \E, P)$ with events of the form $E_A = \{ (x,y) : p(y) - f(x) \in A \}$ with $A$ Borel, and distribution $P(E_A) = \psi(A)$. The fibre of the system is $\D=\{ (x,y) : f(x) = p(y)\}$.
\end{defi} 

\noindent The idea is that morphisms can be composed by interconnecting their names. The precise statement goes as follows:

\begin{thm}\label{thm:interconnection_formula}
Let $(f,\psi,p) : X \to Y$ and $(g,\chi,q) : Y \to Z$ be composable morphisms in $\gaussex$. We form the \emph{name} Gaussian systems $(X \times Y, \E,P)$ and $(Y \times Z, \F,Q)$ and weaken them to systems $(X \times Y \times Z, \E',P')$ and $(X \times Y \times Z, \F',Q')$ defined on the joint space $X \times Y \times Z$ (see Example~\ref{ex:weakening}). The resulting systems have fibres
\begin{align*}
\D_{P'} &= \{ (x,y,z) : f(x) = p(y) \} \\
\D_{Q'} &= \{ (x,y,z) : g(y) = q(z) \}
\end{align*}
We claim that these systems are complementary, and their interconnection is equal to the name of the composite $f ; \cpy_Y ; (\id \otimes g)$. In particular, we can obtain the composite $f;g$ by interconnecting and then marginalizing out the middle wire. 
\end{thm}
\begin{proof}
In Appendix~\ref{app:interconnection}.
\end{proof}

\noindent Note that in a hypergraph category where interconnection is a primitive construct, Theorem~\ref{thm:interconnection_formula} is provable from the axioms and can be rendered as follows graphically. 

\[ \begin{tikzpicture}[scale=0.4]
	\begin{pgfonlayer}{nodelayer}
		\node [style=morphism] (11) at (-2, -4) {$g$};
		\node [style=none] (12) at (-3.5, -4) {};
		\node [style=none] (13) at (0.5, -4) {};
		\node [style=bn] (14) at (-5.5, -3.5) {};
		\node [style=bn] (15) at (-6.5, -3.5) {};
		\node [style=none] (16) at (-3.5, -3) {};
		\node [style=none] (17) at (0.5, -3) {};
		\node [style=bn] (18) at (-6.5, 2) {};
		\node [style=bn] (19) at (-6.5, -2) {};
		\node [style=none] (20) at (0.5, 2) {};
		\node [style=none] (21) at (0.5, -2) {};
		\node [style=morphism] (22) at (-2, 3) {$f$};
		\node [style=none] (23) at (-3.5, 3) {};
		\node [style=none] (24) at (0.5, 3) {};
		\node [style=bn] (25) at (-5.5, 3.5) {};
		\node [style=bn] (26) at (-6.5, 3.5) {};
		\node [style=none] (27) at (-3.5, 4) {};
		\node [style=none] (28) at (0.5, 4) {};
		\node [style=bn] (29) at (4, 2.5) {};
		\node [style=bn] (30) at (4, -2.5) {};
		\node [style=bn] (31) at (4, 0) {};
		\node [style=none] (32) at (6, 2.5) {};
		\node [style=none] (33) at (6, 0) {};
		\node [style=none] (34) at (6, -2.5) {};
		\node [style=none] (35) at (8, 0) {=};
		\node [style=morphism] (36) at (13, -2) {$g$};
		\node [style=none] (37) at (13, -2) {};
		\node [style=none] (38) at (14, -2) {};
		\node [style=bn] (39) at (11, -1.5) {};
		\node [style=bn] (40) at (10, -1.5) {};
		\node [style=none] (41) at (13, -1) {};
		\node [style=none] (42) at (14, -1) {};
		\node [style=morphism] (47) at (13, 1) {$f$};
		\node [style=none] (48) at (13, 1) {};
		\node [style=none] (49) at (14, 1) {};
		\node [style=bn] (50) at (11, 1.5) {};
		\node [style=bn] (51) at (10, 1.5) {};
		\node [style=none] (52) at (13, 2) {};
		\node [style=none] (53) at (14, 2) {};
		\node [style=none] (54) at (16.25, 2) {};
		\node [style=none] (55) at (16.5, -2) {};
		\node [style=bn] (56) at (16.25, 0) {};
		\node [style=none] (57) at (18, 2) {};
		\node [style=none] (58) at (18, 0) {};
		\node [style=none] (59) at (18, -2) {};
		\node [style=none] (60) at (19, 0) {=};
		\node [style=morphism] (61) at (27, -2.5) {$g$};
		\node [style=none] (62) at (27, -2.5) {};
		\node [style=none] (63) at (28, -2.5) {};
		\node [style=bn] (64) at (25, -1.5) {};
		\node [style=morphism] (68) at (23, -1.5) {$f$};
		\node [style=none] (69) at (23, -1.5) {};
		\node [style=bn] (71) at (21, 0) {};
		\node [style=bn] (72) at (20, 0) {};
		\node [style=none] (73) at (23, 1.5) {};
		\node [style=none] (74) at (24, 1.5) {};
		\node [style=none] (75) at (26.25, 1.5) {};
		\node [style=none] (78) at (28, 1.5) {};
		\node [style=none] (79) at (27, -0.5) {};
		\node [style=none] (80) at (28, -0.5) {};
		\node [style=none] (81) at (-0.5, 1) {};
		\node [style=none] (82) at (-7.75, 1) {};
		\node [style=none] (83) at (-7.75, 5) {};
		\node [style=none] (84) at (-0.5, 5) {};
		\node [style=none] (85) at (-0.5, -5) {};
		\node [style=none] (86) at (-7.75, -5) {};
		\node [style=none] (87) at (-7.75, -1) {};
		\node [style=none] (88) at (-0.5, -1) {};
	\end{pgfonlayer}
	\begin{pgfonlayer}{edgelayer}
		\draw (13.center) to (11);
		\draw (11) to (12.center);
		\draw (17.center) to (16.center);
		\draw [in=90, out=-180, looseness=0.50] (16.center) to (14);
		\draw (14) to (15);
		\draw [in=-180, out=-90, looseness=0.50] (14) to (12.center);
		\draw (18) to (20.center);
		\draw (19) to (21.center);
		\draw (24.center) to (22);
		\draw (22) to (23.center);
		\draw (28.center) to (27.center);
		\draw [in=90, out=-180, looseness=0.50] (27.center) to (25);
		\draw (25) to (26);
		\draw [in=-180, out=-90, looseness=0.50] (25) to (23.center);
		\draw [in=225, out=0, looseness=0.50] (21.center) to (29);
		\draw [in=165, out=0] (28.center) to (29);
		\draw [in=150, out=0, looseness=0.75] (24.center) to (31);
		\draw [in=210, out=0, looseness=0.50] (17.center) to (31);
		\draw [in=135, out=0, looseness=0.50] (20.center) to (30);
		\draw [in=-135, out=0, looseness=0.75] (13.center) to (30);
		\draw (29) to (32.center);
		\draw (31) to (33.center);
		\draw (30) to (34.center);
		\draw (38.center) to (36);
		\draw (36) to (37.center);
		\draw (42.center) to (41.center);
		\draw [in=90, out=-180, looseness=0.75] (41.center) to (39);
		\draw (39) to (40);
		\draw [in=-180, out=-90, looseness=0.75] (39) to (37.center);
		\draw (49.center) to (47);
		\draw (47) to (48.center);
		\draw (53.center) to (52.center);
		\draw [in=90, out=-180, looseness=0.75] (52.center) to (50);
		\draw (50) to (51);
		\draw [in=-180, out=-90, looseness=0.75] (50) to (48.center);
		\draw (53.center) to (54.center);
		\draw [in=150, out=0, looseness=0.75] (49.center) to (56);
		\draw [in=210, out=0, looseness=0.50] (42.center) to (56);
		\draw (38.center) to (55.center);
		\draw (54.center) to (57.center);
		\draw (56) to (58.center);
		\draw (55.center) to (59.center);
		\draw (63.center) to (61);
		\draw (61) to (62.center);
		\draw [in=-180, out=-90, looseness=1.25] (64) to (62.center);
		\draw (68) to (69.center);
		\draw (74.center) to (73.center);
		\draw [in=90, out=-180, looseness=0.75] (73.center) to (71);
		\draw (71) to (72);
		\draw [in=-180, out=-90, looseness=0.75] (71) to (69.center);
		\draw (74.center) to (75.center);
		\draw (75.center) to (78.center);
		\draw (69.center) to (64);
		\draw [in=-180, out=90, looseness=1.25] (64) to (79.center);
		\draw (79.center) to (80.center);
		\draw [style=dashed box] (84.center) to (81.center);
		\draw [style=dashed box] (81.center) to (82.center);
		\draw [style=dashed box] (82.center) to (83.center);
		\draw [style=dashed box] (83.center) to (84.center);
		\draw [style=dashed box] (88.center) to (85.center);
		\draw [style=dashed box] (85.center) to (86.center);
		\draw [style=dashed box] (86.center) to (87.center);
		\draw [style=dashed box] (87.center) to (88.center);
	\end{pgfonlayer}
\end{tikzpicture} \]

\noindent Note the weakened name system in the dashed boxes, and their interconnection using $
	\tikzset{baseline=-0.65ex}
	\begin{tikzpicture}[scale=0.1]
	\begin{pgfonlayer}{nodelayer}
		\node [style=none] (0) at (0, 2) {};
		\node [style=bn] (1) at (2, 0) {};
		\node [style=none] (2) at (0, -2) {};
		\node [style=none] (3) at (5, 0) {};
		\node [style=none] (4) at (-2, 2) {};
		\node [style=none] (5) at (-2, -2) {};
	\end{pgfonlayer}
	\begin{pgfonlayer}{edgelayer}
		\draw (0.center) to (4.center);
		\draw [in=90, out=0] (0.center) to (1);
		\draw (5.center) to (2.center);
		\draw [in=-90, out=0] (2.center) to (1);
		\draw (1) to (3.center);
	\end{pgfonlayer}
\end{tikzpicture}}
$. $\gaussex$ is \emph{not} a hypergraph category, but we can still interconnect systems by hand if we check (as in Theorem~\ref{thm:interconnection_formula}) that the relevant complementarity condition holds. 

It is interesting that our use of names is reminiscent of the use of \emph{copy-composition} of \cite{smithe2024copy}, the composition of couplings and \emph{conditional products} (due to \cite{dawid1999conditional}, see also \cite[Section~12]{fritz}, \cite{flori2013compositories}). Lastly, we give an outlook on non-complementary interconnections.

\subsection{Beyond Complementarity}\label{sec:conditioning}

In order to define non-complementary interconnections, we need to develop a theory of conditioning. For example, consider two Gaussian distribution $\psi_1, \psi_2$ on $\R$. What should the following interconnection $\psi$ be?

\begin{equation} \begin{tikzpicture}[scale=0.2]
	\begin{pgfonlayer}{nodelayer}
		\node [style=none] (3) at (1, -3) {};
		\node [style=bn] (5) at (4, 0) {};
		\node [style=none] (6) at (1, 3) {};
		\node [style=none] (7) at (8, 0) {};
		\node [style=state] (8) at (-2, 3) {$\psi_1$};
		\node [style=state] (9) at (-2, -3) {$\psi_2$};
	\end{pgfonlayer}
	\begin{pgfonlayer}{edgelayer}
		\draw (7.center) to (5);
		\draw [in=-90, out=0, looseness=1.25] (3.center) to (5);
		\draw [in=0, out=90, looseness=1.25] (5) to (6.center);
		\draw (8) to (6.center);
		\draw (9) to (3.center);
	\end{pgfonlayer}
\end{tikzpicture} \label{eq:interconnection_cond} \end{equation}

\noindent The naive formula $\psi(A) = \psi_1(A)\psi_2(A)$ does not define a probability measure. In statistical terms, we think of variables $X_1 \sim \psi_1, X_2 \sim \psi_2$, and condition them on equality. For example, if $X_1, X_2 \sim \N(0,1)$, then conditioned on equality we obtain $X_1 \sim \N(0,0.5)$. Conditioning on equality is generally a fragile notion (this has been termed \emph{Borel's paradox}, see e.g. \cite{shan_ramsey,2021compositional,jacobs:paradoxes}). For Gaussians, the situation is unambiguous and has been worked out in \cite{stein2024graphical}. \\

\noindent The first step in any theory of conditioning is the axiomatic notion of \emph{conditional} in a Markov category. This is in essence a powerful way to restructure dataflow in a Markov category. 

\begin{defiC}[{\cite[Section~11]{fritz}}]\label{def:conditional}
A Markov category $\C$ \emph{has conditionals} if for every morphism $f : A \to X \otimes Y$ there exists $f|_X : A \otimes X \to Y$ such that
\[ \begin{tikzpicture}[scale=0.3]
	\begin{pgfonlayer}{nodelayer}
		\node [style=morphism] (0) at (-12, 6) {$f$};
		\node [style=none] (12) at (-11.5, 6.5) {};
		\node [style=none] (13) at (-11.5, 5.5) {};
		\node [style=none] (14) at (-9, 6.5) {};
		\node [style=none] (15) at (-9, 5.5) {};
		\node [style=none] (25) at (-6, 6) {=};
		\node [style=none] (26) at (-15.5, 6) {};
		\node [style=none] (27) at (-12.5, 6) {};
		\node [style=morphism] (28) at (1, 7) {$f$};
		\node [style=none] (29) at (1.5, 7.5) {};
		\node [style=none] (30) at (1.5, 6.5) {};
		\node [style=bn] (31) at (4, 7.5) {};
		\node [style=bn] (32) at (2.5, 6.5) {};
		\node [style=bn] (33) at (-1, 6) {};
		\node [style=none] (34) at (0.75, 7) {};
		\node [style=morphism] (35) at (7.5, 6.5) {$f|_X$};
		\node [style=none] (36) at (5, 8.5) {};
		\node [style=none] (37) at (10, 6.5) {};
		\node [style=none] (38) at (10, 8.5) {};
		\node [style=none] (39) at (-4, 6) {};
		\node [style=none] (40) at (0.5, 5) {};
		\node [style=none] (41) at (3.5, 5) {};
		\node [style=none] (42) at (5.5, 6.25) {};
		\node [style=none] (43) at (5.5, 6.75) {};
		\node [style=none] (44) at (7.25, 6.25) {};
		\node [style=none] (45) at (7.25, 6.75) {};
	\end{pgfonlayer}
	\begin{pgfonlayer}{edgelayer}
		\draw (13.center) to (15.center);
		\draw (12.center) to (14.center);
		\draw (26.center) to (27.center);
		\draw (30.center) to (32);
		\draw (29.center) to (31);
		\draw [in=180, out=90, looseness=1.25] (33) to (34.center);
		\draw [in=180, out=90, looseness=1.50] (31) to (36.center);
		\draw (38.center) to (36.center);
		\draw (37.center) to (35);
		\draw (39.center) to (33);
		\draw [in=180, out=-90, looseness=1.25] (33) to (40.center);
		\draw (40.center) to (41.center);
		\draw [in=180, out=0] (41.center) to (42.center);
		\draw [in=-90, out=180] (43.center) to (31);
		\draw (43.center) to (45.center);
		\draw (42.center) to (44.center);
	\end{pgfonlayer}
\end{tikzpicture} \]
\end{defiC}

\begin{prop}
$\gaussex$ has conditionals.
\end{prop}
\begin{proof}
Both $\gauss$ and $\tlinrel$ have conditionals, and conditionals in $\gaussex$ can be constructed from to these two special cases. The full proof is given in \cite{stein2023category}.
\end{proof}

\noindent We give an example to show how conditionals relate to interconnection: The noisy resistor was described by the joint distribution $P = P_{VI}$ given in the following form 
\begin{align*}
I &\sim \R \\
V &\sim \N(RI,\sigma^2)
\end{align*}
In order to condition on $V$, Definition~\ref{def:conditional} asks us to rewrite (or factorize) this joint distribution distribution in a way that samples $V$ ``first'', and chooses $I$ depending on $V$. Using the reasoning of \eqref{eq:interconnection}, the following is a valid conditional factorization
\begin{align*}
V &\sim \R \\
I &\sim \N(V/R,\sigma^2/R^2)
\end{align*}
We can render this factorization in string diagrams,
\[ \begin{tikzpicture}[scale=0.2]
	\begin{pgfonlayer}{nodelayer}
		\node [style=bn] (0) at (-31, -2) {};
		\node [style=bn] (1) at (-26, -2) {};
		\node [style=none] (2) at (-24, -4) {};
		\node [style=none] (3) at (-24, 0) {};
		\node [style=morphism] (4) at (-16, 0) {$\times R$};
		\node [style=none] (5) at (-10, 0) {};
		\node [style=wn] (6) at (-8, 2) {};
		\node [style=none] (7) at (-10, 4) {};
		\node [style=morphism] (8) at (-16, 4) {$\N(0,\sigma^2)$};
		\node [style=none] (9) at (-3, 2) {};
		\node [style=none] (10) at (-3, -4) {};
		\node [style=bn] (11) at (6, 2) {};
		\node [style=bn] (12) at (11, 2) {};
		\node [style=none] (13) at (13, 4) {};
		\node [style=none] (14) at (13, 0) {};
		\node [style=morphism] (15) at (22, 0) {$\times R^{-1}$};
		\node [style=none] (16) at (28, 0) {};
		\node [style=wn] (17) at (30, -2) {};
		\node [style=none] (18) at (28, -4) {};
		\node [style=morphism] (19) at (22, -4) {$\N(0,\sigma^2/R^2)$};
		\node [style=none] (20) at (38, -2) {};
		\node [style=none] (21) at (38, 4) {};
		\node [style=none] (23) at (2, 0) {=};
		\node [style=none] (24) at (-6, -2) {};
		\node [style=none] (25) at (-6, 7) {};
		\node [style=none] (26) at (-23, 7) {};
		\node [style=none] (27) at (-23, -2) {};
		\node [style=none] (28) at (-1.5, -4) {$I$};
		\node [style=none] (29) at (-29, -5) {$I$};
		\node [style=none] (30) at (-1.5, 2) {$V$};
		\node [style=none] (31) at (9, 5) {$V$};
		\node [style=none] (32) at (39.5, 4) {$V$};
		\node [style=none] (33) at (39.5, -2) {$I$};
		\node [style=none] (34) at (34, -7) {};
		\node [style=none] (35) at (34, 2.5) {};
		\node [style=none] (36) at (14, 2.5) {};
		\node [style=none] (37) at (14, -7) {};
	\end{pgfonlayer}
	\begin{pgfonlayer}{edgelayer}
		\draw (0) to (1);
		\draw [in=-180, out=-45] (1) to (2.center);
		\draw (2.center) to (10.center);
		\draw [in=-180, out=45] (1) to (3.center);
		\draw (3.center) to (4);
		\draw (4) to (5.center);
		\draw [in=-135, out=0, looseness=0.75] (5.center) to (6);
		\draw [in=0, out=135] (6) to (7.center);
		\draw (7.center) to (8);
		\draw (6) to (9.center);
		\draw (11) to (12);
		\draw [in=180, out=45] (12) to (13.center);
		\draw (13.center) to (21.center);
		\draw [in=180, out=-45] (12) to (14.center);
		\draw (14.center) to (15);
		\draw (15) to (16.center);
		\draw [in=135, out=0, looseness=0.75] (16.center) to (17);
		\draw [in=0, out=-135] (17) to (18.center);
		\draw (18.center) to (19);
		\draw (17) to (20.center);
		\draw [style=dashed box] (27.center) to (24.center);
		\draw [style=dashed box] (24.center) to (25.center);
		\draw [style=dashed box] (25.center) to (26.center);
		\draw [style=dashed box] (26.center) to (27.center);
		\draw [style=dashed box] (37.center) to (34.center);
		\draw [style=dashed box] (34.center) to (35.center);
		\draw [style=dashed box] (35.center) to (36.center);
		\draw [style=dashed box] (36.center) to (37.center);
	\end{pgfonlayer}
\end{tikzpicture} \]
On the left hand side (see Example~\ref{ex:stringdiag}), the $V$-wire is computed depending on $I$, while on the right hand side this process is inverted, and $V$ is sampled first. The dashed boxes indicate the conditional morphisms $P|_I$ and $P|_V$ respectively.

\pagebreak

\section{Related Work and Wider Context}\label{sec:connections}

In this section, we give some wider theoretical context for our construction. In \ref{sec:cmp} we give a brief overview over categorical approaches to probability. In \ref{sec:wquot}, we systematically address the role of coarse $\sigma$-algebras and the question of equivalence of systems. \\

Sections~\ref{sec:duality}~and~\ref{sec:presentations} portray different perspectives on extended Gaussians from recent follow-up work \cite{stein2023towards,stein2024graphical}. Those include duality considerations involving convex analysis and quadratic functions, and an algebraic presentation of the category $\gaussex$ in the style of Graphical Linear Algebra. We conclude the section by discussing a prototype implementation in \ref{sec:implementation}.

\subsection{Categorical Models of Probability Theory}\label{sec:cmp}

Markov categories (and related formalisms such as copy-delete and hypergraph categories) are one widely used axiomatization of what it means to be a categorical model of probability. They focus on the monoidal structure of stochastic computation (parallel vs sequential composition, correlation vs independence). We have used the language of Markov categories here because our notion of extended Gaussian morphisms exhibits this kind of structure and satisfies the relevant axioms. This type of categorical model also fits well with call-by-value semantics of probabilistic programming languages \cite{staton:commutative_semantics,stein2021structural}.

There are other approaches to categorical probability with different goals in mind. We give a brief and non-exhaustive overview, and comment on the extent that they might relate to open stochastic systems.

Probability theory exhibits rich dualities between measure-theoretic and functional-analytic structures such as C*-algebras or von Neumann algebras (known as Gelfand duality). This approach is the basis for generalizations of probability and topology such as noncommutative geometry, free and quantum probability, and has seen extensive categorical treatment (e.g. \cite{furber2015kleisli,pavlov2022gelfand}). Extended Gaussians fit well into this picture, and we will mention some of their duality theory in the remainder of this section. 

Effectus theory \cite{cho2015introduction} is built around properties of the coproduct $+$ and a calculus of effects rather than the tensor $\otimes$ of stochastic computation, and encompasses probabilistic as well as quantum computation. It has a duality theory in terms of effect algebras generalizing the previously mentioned dualities. However, the categories $\gauss$ and $\gaussex$ don't have coproducts and therefore fail to be effectuses. Similarly, the study of probability monads (e.g. \cite{giry,avery2016codensity,fritz2017probability}) seems too general for our fairly specific setup, and combining monads for probability and nondeterminism is tricky as discussed in the introduction.

Other areas in categorical probability theory are interested in a coalgebraic view of systems \cite{sokolova2011probabilistic}, higher-order probabilistic computation such as quasi-Borel spaces \cite{heunen:qbs,omegaqbs}, probability sheaves and separation logic \cite{simpson-sheaves,simpson2024equivalence,li2024nominal}, linear logic \cite{dahlqvist:roban,pcoh}, or lenses and fibrational structure \cite{smithe2024copy,smithe:bayesian,smithe2023approximate}. Connections with our work remain to be explored.

\subsection{\texorpdfstring{$\sigma$}{σ}-algebras and Weak Quotients}\label{sec:wquot}

Intuitively, a Gaussian system on $\R^n$ with fibre $\D$ is given by a Gaussian distribution on a quotient vector space $\R^n/\D$. The quotient space is isomorphic to $\R^k$ where $k=n-\dim(\D)$, so we could explain a Borel structure on the quotient and talk about Gaussian distributions directly. However, Willems' definition uses the space $\R^n \wquot \D$, which is $\R^n$ equipped with the coarse $\sigma$-algebra of Borel cylinders $\B_{\D}(\R^n)$. What is the relationship between these two kinds of quotient? \\

\noindent We have a measurable projection $p : \R^n \wquot D \to \R^n/\D$. This map fails to have an inverse in the category of measurable spaces $\meas$, because on underlying sets, $\R^n \to \R^n/\D$ is not injective. The induced channel $\lift p : \R^n \wquot \D \chanto \R^n/\D$ on the other hand is invertible! We write its inverse as  $p^{\sim 1} : \R^n/\D \chanto \R^n \wquot \D$, and define it as follows:
\begin{equation} p^{\sim 1} : (\R^n/\D) \times \B_\D(\R^n) \to [0,1], \quad p^{\sim 1}(A|[x]) = [x \in A] \label{eq:def_pinv} \end{equation}
Note that the value of $p^{\sim 1}(A|[x])$ is well-defined precisely because we only evaluate it on Borel cylinders parallel to $\D$. \\

\noindent That is, in the category $\stoch$, it does not matter whether we work with $\R^n/\D$ or $\R^n\wquot \D$. On the level of measurable functions however, there is a difference. This is reminiscent of point-free topology, where there are maps $\mathcal O_Y \to \mathcal O_X$ between frames that do not come from any continuous map between topological spaces $X \to Y$. We call the construction $\R^n \wquot \D$ a \emph{weak quotient} in analogy to this, following the discussion of Moss and Perrone in \cite[Appendix A]{moss_perrone_2023}. Our discussion fits their framework, with $\B_\D(\R^n)$ being an instance of an \emph{invariant $\sigma$-algebra} for the translations in $\D$. \\

\noindent The difference between $\meas$ and $\stoch$ also concerns a problem which has remained largely implicit in Willems: \emph{What is the correct notion of equivalence of systems?}

If we work with channels, we work point-free and need not bother with weak quotients. On the other hand, the distinction between open and closed system is not stable under isomorphism in $\stoch$. The space $\R^n \wquot \D$ is open, but isomorphic to $\R^n/\D$, which is Borel (closed). We will now recall some of the categorical structure for isomorphisms in $\stoch$: 

A channel $\kappa : (X,\E) \chanto (Y,\F)$ is called \emph{deterministic} if it is 0/1-valued, i.e. $\kappa(F|x) \in \{0,1\}$ for all $x \in X, F \in \F$. A channel $\kappa$ is deterministic if and only if it is deterministic in the sense of Definition~\ref{def:determinism}, i.e. commutes with copying \cite{fritz}.	It is easy to see that composition and tensor of deterministic channels are deterministic, and deterministic channels thus form a Markov subcategory $\stoch_\det$ of $\stoch$. $\stoch_\det$ is \emph{cartesian monoidal}, that is the tensor forms a categorical product. 

We are interested in determinism because of the following key property: 

\begin{propC}[{\cite[11.28]{fritz}}]
Every isomorphism in $\stoch$ is  deterministic. 
\end{propC}

\noindent Every channel $\lift f$ arising from a measurable function is deterministic, but the converse is false, as the morphism $p^{\sim 1}$ in \eqref{eq:def_pinv} shows. Instead, we use the following definition:

\begin{defiC}[{\cite{moss2022probability}}]
A measurable space $(Y,\F)$ is called \emph{sober} if every deterministic channel $(X,\E) \chanto (Y,\F)$ comes from a unique measurable function $(X,\E) \to (Y,\F)$.
\end{defiC}
\noindent A sufficient condition for sobriety is that $\E$ separate points and be countably generated \cite[10.4]{fritz}. In particular, all Borel spaces are sober. The space $\R^n \wquot \D$ is sober if and only if $\D = 0$. \\

\noindent Every function $f : X \to Y$ induces a pushforward and pullback of $\sigma$-algebras; namely for $\sigma$-algebras $\E,\F$ on $X,Y$ respectively, we define
\begin{align*}
	f_*\E &= \{ F \subseteq Y \s f^{-1}(F) \in E \} \\
	f^{-1}[\F] &= \{ f^{-1}(F) \s F \in F \}
\end{align*}
For any $\E,\F$ we have the adjoint relationship $f^{-1}[\F] \subseteq \E \Leftrightarrow \F \subseteq f_*\E$, and $f$ is measurable $(X,\E) \to (Y,\F)$ precisely when the equivalent conditions hold. In this case, the sub-$\sigma$-algebra $f^{-1}[\F] \subseteq \E$ is often called the $\sigma$-algebra generated by $f$. 

\begin{propC}[{\cite{ensarguet2024categorical}}]
	Let $f : (X,\E) \to (Y,\F)$ be any measurable, surjective function and $f^{-1}[\F]$ the induced $\sigma$-algebra on $X$. Then the induced channel $\tilde f : (X,f^{-1}[\F]) \leadsto (Y,\F)$ given by $\tilde f(F|x) = [f(x) \in F]$ is an isomorphism in $\stoch$. 
\end{propC}

\noindent This proposition illuminates the construction of the pullback measure: 

\begin{enumerate}
\item For every inclusion of $\sigma$-algebras $E' \subseteq \E$, the identity defines a measurable function $\pi : (X,\E) \to (X,\E')$. This \emph{coarse-graining} is an epimorphism in both $\meas$ and $\stoch$.
\item Now, given any measurable surjection $p : (X,\E) \twoheadrightarrow (Y,\F)$, we can identify $(Y,\F)$ with $(X,p^{-1}[\E])$. Given a measure $\psi$ on $(Y,\F)$, taking the pullback measure $p^{-1}(\psi)$ simply corresponds to composing with the isomorphism. We can capture the whole situation by a commuting diagram in $\stoch$
\[\begin{tikzcd}
	& {(Y,\F)} \\
	I && {(X,\E)} \\
	& {(X,p^{-1}[\F])}
	\arrow["p"', two heads, from=2-3, to=1-2]
	\arrow["\pi", two heads, from=2-3, to=3-2]
	\arrow["\psi", from=2-1, to=1-2]
	\arrow["{p^{-1}(\psi)}"', from=2-1, to=3-2]
	\arrow["\sim"{description}, from=1-2, to=3-2]
\end{tikzcd}\]
\end{enumerate}

\noindent This suggests that when trying to define a general formalism of channels in which the codomain $\sigma$-algebra can vary, we can use copartial maps in $\stoch$, whose right legs are surjective measurable maps. Checking the details (e.g. if the relevant pushouts exist) is work for the future. \\

\noindent This represents a starting point for extending our current approach to nonlinear open systems. For example, Willems gives the example of supply and demand curves (see Figure~\ref{fig:nonlinear}).

\begin{figure}[h] 
	\includegraphics[scale=0.2]{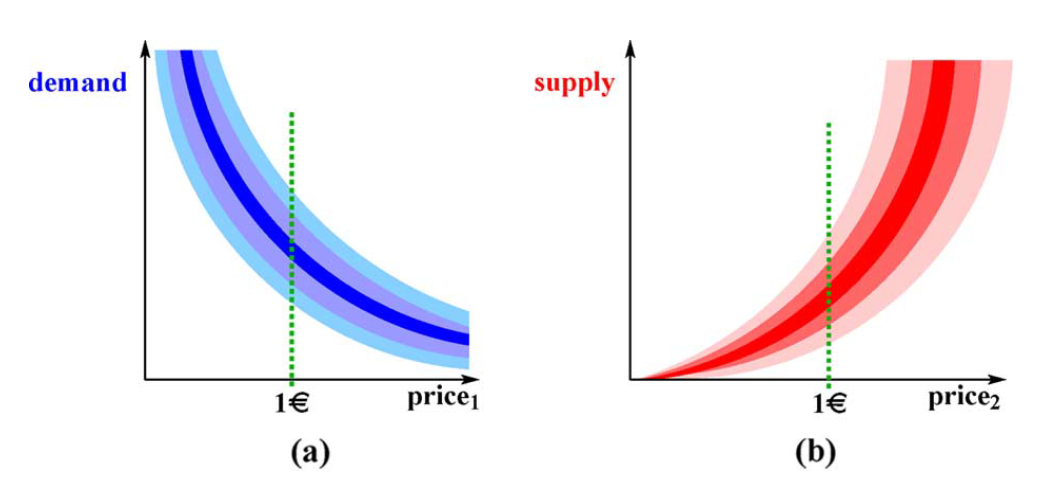}
	\caption{Supply and demand as nonlinear open systems (Figure 17 of \cite{willems:oss})}
	\label{fig:nonlinear}
\end{figure}

\subsection{Variance-Precision Duality and Quadratic Functions}\label{sec:duality}

So far, we have described extended Gaussian distributions as categorical and measure-theoretic objects. We will briefly sketch how they naturally arise naturally from an analytic point of view. For simplicity, we will only consider \emph{centered Gaussians} of mean zero. 

A Gaussian distribution with nonsingular, i.e. positive definite, covariance matrix $\Sigma \in \R^{n \times n}$ has a density function over $\R^n$ of the form
\begin{equation} f(x) \propto \exp\left(-\frac 1 2 x^T\Omega x\right) \label{eq:density} \end{equation}
The matrix $\Omega$ appearing in the density, known as the \emph{precision} or \emph{information matrix}, is the inverse of the covariance matrix $\Omega = \Sigma^{-1}$. Nonsingular Gaussian distributions thus have two equivalent representations, a \emph{precision form} and a \emph{covariance form}, which are related to each other by inversion. The two forms each have convenient properties for different applications
\begin{enumerate}
\item covariance is additive for convolution, i.e. we have
\[ \Sigma_{\psi_1 + \psi_2} = \Sigma_{\psi_1} + \Sigma_{\psi_2} \]
because variances of independent variables add. 
\item precision is additive for conditioning. That is we have \[ \Omega_{\psi_1 \bullet \psi_2} = \Omega_{\psi_1} + \Omega_{\psi_2} \] where $\psi_1 \bullet \psi_2$ denotes the interconnection of \eqref{eq:interconnection_cond}. That is because when conditioning, densities multiply, so precisions add. 
\end{enumerate}

\begin{figure}
	\includegraphics[scale=0.30]{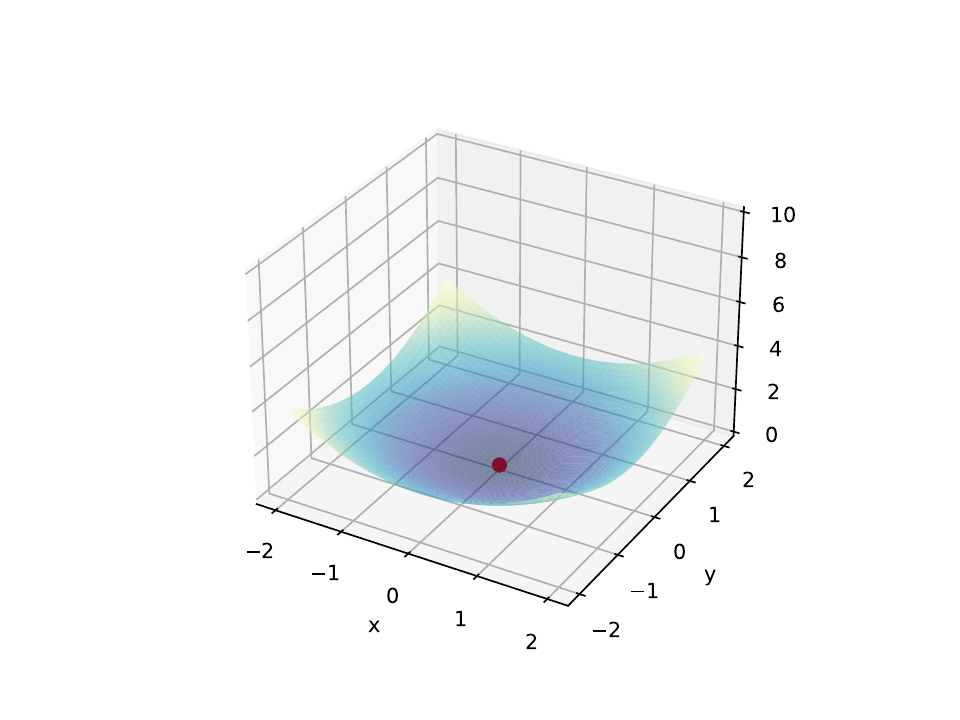}
	\includegraphics[scale=0.30]{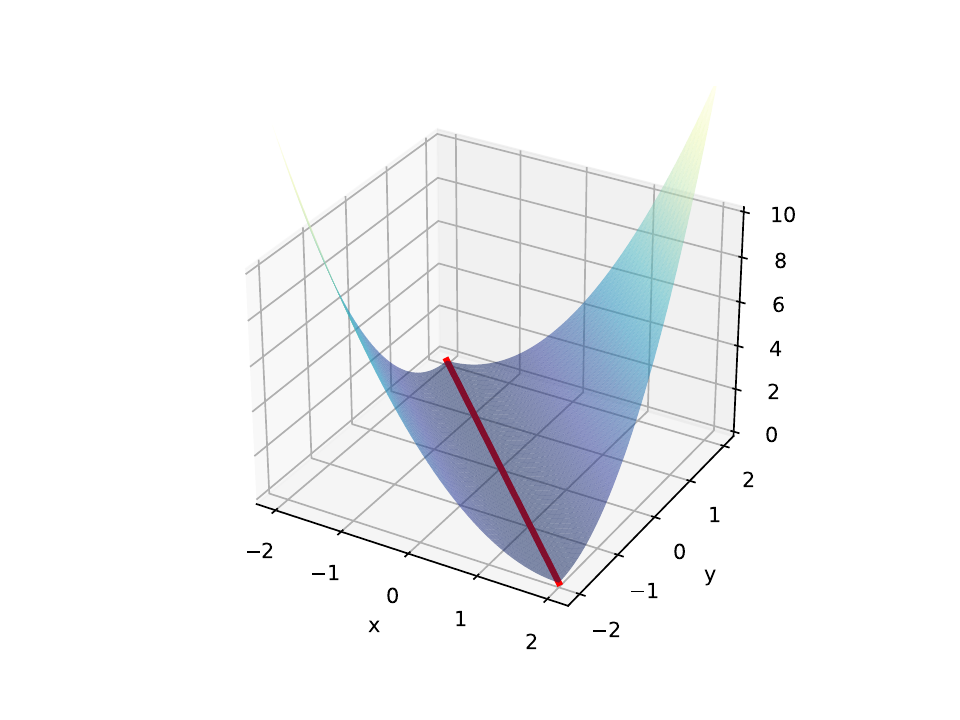}
	\includegraphics[scale=0.30]{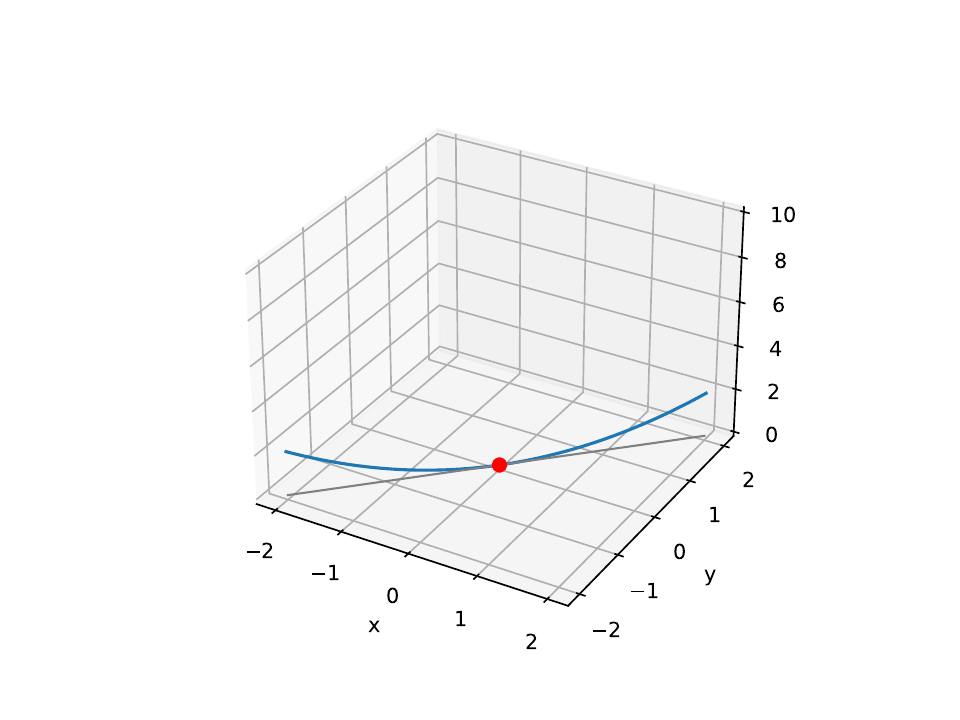}
	\caption{Graphs of partial quadratic functions on $\R^2$ (kernel $\D$ shown in red)}
	\label{fig:pqf}
\end{figure}

\noindent The duality between covariance and precision form is lost once we consider Gaussians whose covariance matrix is singular. Such distribution no longer admit a density over all of $\R^n$.
\begin{enumerate}
\item What is the precision form of a singular Gaussian distribution?
\item What would be the meaning of a singular precision matrix?
\end{enumerate} 

\noindent The answer to the second question is precisely extended Gaussian distributions, whose fibre is $\D=\ker(\Omega)$. Under this viewpoint, they have been studied for example in \cite{james:variance_manifold}. To understand this duality, it is useful to represent Gaussian distributions as quadratic functions: \\

\noindent Let $\chi$ be a centered extended Gaussian distribution with fibre $\D$ and covariance matrix $\Sigma$, and let $\S = \im(\Sigma) + \D$ be its support. Then we associate to $\chi$ two partial quadratic functions, a precision function $f_p$ and covariance function $f_c : \R^m \to [0,+\infty]$, as follows:

\begin{align*}
f_p(x) = \frac 1 2 x^T \Sigma^+ x + \iv{x \in \S}, \qquad f_c(x) = \frac 1 2 x^T \Sigma x + \iv{x \in \D^\bot}  
\end{align*}

\noindent Here, the matrix $\Sigma^+$ is the Moore-Penrose pseudoinverse of $\Sigma$, and $\D^\bot$ denotes the orthogonal complement of the subspace $\D$. We abbreviate $\iv{\phi} = 0$ if $\phi$ is true, and $+\infty$ otherwise, that is we use $+\infty$ to indicate partiality. We note that $f_p$ is reminiscent of a negative log-density \eqref{eq:density} if that density exists. The properties of the functions $f_p, f_c$ can be summarized as follows:

\begin{center}
	\begin{tabular}{ c|cc } 
			& precision form $f_p$ & covariance form $f_c$ \\ \hline
			domain & $\S$ & $\D^\bot$ \\ \hline
			kernel & $\D$ & $\S^\bot$ \\
		\end{tabular}
\end{center}
By \emph{domain} we mean the set of points where $f(x)<\infty$. By kernel, we mean a subspace of degeneracy or translation invariance, i.e. $f(x) = f(x+d)$ for all $d \in \D$ (see Figure~\ref{fig:pqf}).

For ordinary Gaussian distributions with $\D = 0$, we obtain the table
\begin{center}
	\begin{tabular}{ c|cc } 
		& precision form $f_p$ & covariance form $f_c$ \\ \hline
		domain & $\S$ & $\R^n$ \\ \hline
		kernel & $0$ & $\S^\bot$ \\
	\end{tabular}
\end{center}
that is, the precision form is partial but nondegenerate. The covariance form is globally defined but has a nontrivial kernel. This explains the suffiency of the covariance matrix $\Sigma$ over the precision matrix $\Sigma$ for ordinary Gaussians. 

On the other hand, for an extended Gaussian distribution with full support, we have
\begin{center}
	\begin{tabular}{ c|cc } 
		& precision form $f_p$ & covariance form $f_c$ \\ \hline
		domain & $\R^n$ & $\D^\bot$ \\ \hline
		kernel & $\D$ & $0$ \\
	\end{tabular}
\end{center}
That is, such distributions have a globally defined (possibly singular) precision form, but only a partial covariance form.

By considering extended Gaussian distributions (and partial quadratic functions), we retain the symmetry between the two representations. The duality relating the two representations is an instance of the famous Legendre-Transform from convex analysis:

\begin{prop}
	The partial quadratic functions $f_p,f_c$ are always convex, and related to each other via convex conjugation (a.k.a. Legendre-Fenchel transform).
\end{prop}

This is elaborated and proved in \cite{stein2023towards,stein2024graphical}.

\subsection{Presentations and Hypergraph Categories}\label{sec:presentations} 

\emph{Graphical Linear Algebra} refers to a set of sound and complete diagrammatic axiomatizations of various categories like linear maps, linear relations and affine relations \cite{graphical_la,bonchi:affine,bonchi:interacting_hopf}. Presentations of this kind have been used to reason about control-theoretic framework like signal flow diagrams, electrical circuits and Willems' linear-time invariant (LTI) systems \cite{bonchi:cat_signal_flow,fong:open_systems}. \\

In \cite{stein2024graphical}, we introduce Graphical Quadratic Algebra (GQA), which is an extension of the language of Graphical Affine Algebra with a generator $
	\tikzset{baseline=-0.65ex}
	\begin{tikzpicture}[scale=0.1]
	\begin{pgfonlayer}{nodelayer}
		\node [style=nn] (1) at (-1, 0) {};
		\node [style=none] (3) at (4, 0) {};
	\end{pgfonlayer}
	\begin{pgfonlayer}{edgelayer}
		\draw (1) to (3.center);
	\end{pgfonlayer}
\end{tikzpicture}
}
 : I \to \R$ representing the distribution $\N(0,1)$ (see Figure~\ref{fig:gqa_signature}). This is sufficient to express any extended Gaussian map as a combination of normal distributions and linear relations. We give a sound and complete equational theory for $\gaussex$. Completeness means that diagrammatic reasoning like Example~\ref{ex:stringdiag} can be used to prove every valid equality between extended Gaussian maps. \\

GQA models a strict superset $\gaussex$; in particular it features connectors like $
	\tikzset{baseline=-0.65ex}
	\begin{tikzpicture}[scale=0.1]
	\begin{pgfonlayer}{nodelayer}
		\node [style=none] (0) at (0, 2) {};
		\node [style=bn] (1) at (2, 0) {};
		\node [style=none] (2) at (0, -2) {};
		\node [style=none] (3) at (5, 0) {};
		\node [style=none] (4) at (-2, 2) {};
		\node [style=none] (5) at (-2, -2) {};
	\end{pgfonlayer}
	\begin{pgfonlayer}{edgelayer}
		\draw (0.center) to (4.center);
		\draw [in=90, out=0] (0.center) to (1);
		\draw (5.center) to (2.center);
		\draw [in=-90, out=0] (2.center) to (1);
		\draw (1) to (3.center);
	\end{pgfonlayer}
\end{tikzpicture}}
$ and thus describes a hypergraph category which can model arbitrary interconnections. We call such an entity a \emph{Gaussian relation}, because it combines Gaussian probability with arbitrary (not just total) linear relations. Partial convex quadratic functions form another model of GQA, making the connection to the dualities in Section~\ref{sec:duality} precise.

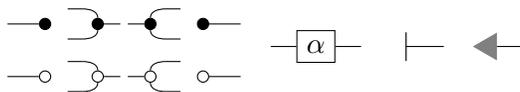
\begin{figure}
	\[ \begin{tikzpicture}[scale=0.1]
	\begin{pgfonlayer}{nodelayer}
		\node [style=none] (0) at (0, 2) {};
		\node [style=bn] (1) at (2, 0) {};
		\node [style=none] (2) at (0, -2) {};
		\node [style=none] (3) at (5, 0) {};
		\node [style=none] (4) at (-2, 2) {};
		\node [style=none] (5) at (-2, -2) {};
		\node [style=bn] (6) at (-5, 0) {};
		\node [style=none] (7) at (-10, 0) {};
		\node [style=none] (8) at (11, 2) {};
		\node [style=bn] (9) at (9, 0) {};
		\node [style=none] (10) at (11, -2) {};
		\node [style=none] (11) at (6, 0) {};
		\node [style=none] (12) at (13, 2) {};
		\node [style=none] (13) at (13, -2) {};
		\node [style=bn] (14) at (16, 0) {};
		\node [style=none] (15) at (21, 0) {};
		\node [style=none] (16) at (0, -5) {};
		\node [style=wn] (17) at (2, -7) {};
		\node [style=none] (18) at (0, -9) {};
		\node [style=none] (19) at (5, -7) {};
		\node [style=none] (20) at (-2, -5) {};
		\node [style=none] (21) at (-2, -9) {};
		\node [style=wn] (22) at (-5, -7) {};
		\node [style=none] (23) at (-10, -7) {};
		\node [style=none] (24) at (11, -5) {};
		\node [style=wn] (25) at (9, -7) {};
		\node [style=none] (26) at (11, -9) {};
		\node [style=none] (27) at (6, -7) {};
		\node [style=none] (28) at (13, -5) {};
		\node [style=none] (29) at (13, -9) {};
		\node [style=wn] (30) at (16, -7) {};
		\node [style=none] (31) at (21, -7) {};
		\node [style=none] (32) at (48, -3) {};
		\node [style=none] (33) at (43, -3) {$|$};
		\node [style=none] (34) at (25, -3) {};
		\node [style=morphism] (35) at (31, -3) {$\alpha$};
		\node [style=none] (36) at (37, -3) {};
		\node [style=nn] (37) at (54, -3) {};
		\node [style=none] (38) at (59, -3) {};
	\end{pgfonlayer}
	\begin{pgfonlayer}{edgelayer}
		\draw (0.center) to (4.center);
		\draw [in=90, out=0] (0.center) to (1);
		\draw (5.center) to (2.center);
		\draw [in=-90, out=0] (2.center) to (1);
		\draw (1) to (3.center);
		\draw (8.center) to (12.center);
		\draw [in=90, out=180] (8.center) to (9);
		\draw (13.center) to (10.center);
		\draw [in=-90, out=180] (10.center) to (9);
		\draw (9) to (11.center);
		\draw (7.center) to (6);
		\draw (15.center) to (14);
		\draw (16.center) to (20.center);
		\draw [in=90, out=0] (16.center) to (17);
		\draw (21.center) to (18.center);
		\draw [in=-90, out=0] (18.center) to (17);
		\draw (17) to (19.center);
		\draw (24.center) to (28.center);
		\draw [in=90, out=180] (24.center) to (25);
		\draw (29.center) to (26.center);
		\draw [in=-90, out=180] (26.center) to (25);
		\draw (25) to (27.center);
		\draw (23.center) to (22);
		\draw (31.center) to (30);
		\draw [in=180, out=0] (33.center) to (32.center);
		\draw (36.center) to (35);
		\draw (35) to (34.center);
		\draw (38.center) to (37);
	\end{pgfonlayer}
\end{tikzpicture} \]
	\caption{The signature of GQA}
	\label{fig:gqa_signature}
\end{figure}

\subsection{Implementation}\label{sec:implementation}

We have given a proof of concept implementation of extended Gaussians in the Julia library \texttt{GaussianRelations.jl} \cite{gaussianrelations}. The library features covariance and precision representations as discussed in Section~\ref{sec:duality}, and implements categorical composition, interconnection and conditioning. \\

\noindent The following code implements the noisy resistor in the style of Example~\ref{ex:noisy_kernel} as a pullback distribution.
\begin{lstlisting}
using GaussianRelations
	
# Construct the noisy resistor with some parameters
	
sigma = 1/4
R     = 1/2

# epsilon = N(0, sigma^2)
epsilon = CovarianceForm(0, sigma^2)

# Define an extended Gaussian via pullback
# [1 -R] * [ V; I ] ~ epsilon
P_VI = [1 -R] \ epsilon
\end{lstlisting}

\noindent It will be interesting to integrate this implementation with probabilistic programming languages like \cite{gaussianinfer}, and categorical frameworks like \texttt{AlgebraicJulia}\footnote{\url{https://github.com/AlgebraicJulia}}.

\bibliographystyle{alphaurl}
\bibliography{main}

\section{Appendix}

\subsection{Copartiality}\label{app:copartiality}

\begin{prop}
The construction $R : \Copar(\set) \to \trel$ is \emph{not} functorial.
\end{prop}
\begin{proof}
We consider sets $X=\{x_1,x_2\}, Y = \{y_1,y_2,y_3,y_4\}$ and $Z=\{z_1,z_2,z_3\}$, and construct copartial maps $f : X \coparto Y, g : Y \coparto Z$ as shown below (the blobs indicate which elements of the codomain are identified).
\[ \begin{tikzpicture}[scale=0.4]
	\begin{pgfonlayer}{nodelayer}
		\node [style=treenode] (0) at (-12, -4) {$x_2$};
		\node [style=treenode] (1) at (-12, 4) {$x_1$};
		\node [style=treenode] (2) at (-5, 6) {$y_1$};
		\node [style=treenode] (3) at (-5, 2) {$y_2$};
		\node [style=treenode] (4) at (-5, -6) {$y_4$};
		\node [style=treenode] (5) at (-5, -2) {$y_3$};
		\node [style=treenode] (6) at (6, 4) {$z_1$};
		\node [style=treenode] (7) at (6, 0) {$z_2$};
		\node [style=treenode] (8) at (6, -4) {$z_3$};
		\node [style=none] (9) at (-4, 1) {};
		\node [style=none] (10) at (-6, 1) {};
		\node [style=none] (11) at (-6, 7) {};
		\node [style=none] (12) at (-4, 7) {};
		\node [style=none] (13) at (-4, -7) {};
		\node [style=none] (14) at (-6, -7) {};
		\node [style=none] (15) at (-6, -1) {};
		\node [style=none] (16) at (-4, -1) {};
		\node [style=treenode] (17) at (3, 6) {$y_1$};
		\node [style=treenode] (18) at (3, 2) {$y_2$};
		\node [style=treenode] (19) at (3, -6) {$y_4$};
		\node [style=treenode] (20) at (3, -2) {$y_3$};
		\node [style=treenode] (21) at (13, -4) {$x_2$};
		\node [style=treenode] (22) at (13, 4) {$x_1$};
		\node [style=treenode] (27) at (20.25, 4) {$z_1$};
		\node [style=treenode] (28) at (20.25, 0) {$z_2$};
		\node [style=treenode] (29) at (20.25, -4) {$z_3$};
		\node [style=none] (30) at (21.25, -5) {};
		\node [style=none] (31) at (19.25, -5) {};
		\node [style=none] (32) at (21.25, 5) {};
		\node [style=none] (33) at (19.25, 5) {};
		\node [style=none] (34) at (-1, 0) {$\large ;$};
		\node [style=none] (35) at (10, 0) {$\large =$};
	\end{pgfonlayer}
	\begin{pgfonlayer}{edgelayer}
		\draw [style=arrow] (1) to (2);
		\draw [style=arrow] (1) to (3);
		\draw [style=arrow] (0) to (5);
		\draw [style=arrow] (0) to (4);
		\draw [style=dashed box] (11.center) to (10.center);
		\draw [style=dashed box] (10.center) to (9.center);
		\draw [style=dashed box] (9.center) to (12.center);
		\draw [style=dashed box] (12.center) to (11.center);
		\draw [style=dashed box] (15.center) to (14.center);
		\draw [style=dashed box] (14.center) to (13.center);
		\draw [style=dashed box] (13.center) to (16.center);
		\draw [style=dashed box] (16.center) to (15.center);
		\draw [style=arrow] (17) to (6);
		\draw [style=arrow] (18) to (7);
		\draw [style=arrow] (20) to (7);
		\draw [style=arrow] (19) to (8);
		\draw [style=arrow] (22) to (27);
		\draw [style=arrow] (22) to (28);
		\draw [style=arrow] (22) to (29);
		\draw [style=arrow] (21) to (27);
		\draw [style=arrow] (21) to (28);
		\draw [style=arrow] (21) to (29);
		\draw [style=dashed box] (31.center) to (30.center);
		\draw [style=dashed box] (30.center) to (32.center);
		\draw [style=dashed box] (32.center) to (33.center);
		\draw [style=dashed box] (33.center) to (31.center);
	\end{pgfonlayer}
\end{tikzpicture} \]
The composite $f;g$ in $\Copar(\set)$ identifies all elements on $Z$. On the other hand, the relational composite $R(f) ; R(g)$ maps $x_1$ to $\{z_1,z_2\}$ but not $z_3$, as would be required for the identity $R(f) ; R(g) = R(f;g)$ to hold.
\end{proof}

\begin{prop}
Define the cartesian product of copartial maps $(X_i \xrightarrow{f_1} P_i \xleftarrow{p_i} Y_i)_{i=1,2}$ as the copartial map $(X_1 \times X_2 \xrightarrow{f_1 \times f_2} P_1 \times P_2 \xleftarrow{p_1 \times p_2})$. Then this construction fails to be bifunctorial, i.e. does not define a monoidal structure on $\Copar(\set)$. 
\end{prop}
\begin{proof}
Let $B$ be a set and $A \subset B$ a proper subset. We write $B/A$ for the set $(B \setminus A) \cup \{\star\}$, which collapses all of $A$ to a point. This is the pushout of the span $1 \xleftarrow{} A \hookrightarrow B$.

Define the copartial maps
\[ u_A = (\emptyset \xrightarrow{} 1 \leftarrow A), \qquad i = (A \hookrightarrow B \xleftarrow{\id} B )\]
Then the composite $(u_A ; i) \times (u_A ; i)$ is the product of pushouts
\[\begin{tikzcd}
	&& {(B/A)^2} \\
	& 1 && {B^2} \\
	\emptyset && {A^2} && {B^2}
	\arrow[from=3-3, to=2-2]
	\arrow[from=3-3, to=2-4]
	\arrow[from=2-4, to=1-3]
	\arrow[from=2-2, to=1-3]
	\arrow[from=3-5, to=2-4]
	\arrow[from=3-1, to=2-2]
\end{tikzcd}\]
while the composite $(u_A \times u_A) ; (i \times i)$ equals the pushout of products
\[\begin{tikzcd}
	&& {B^2/A^2} \\
	& 1 && {B^2} \\
	\emptyset && {A^2} && {B^2}
	\arrow[from=3-3, to=2-2]
	\arrow[from=3-3, to=2-4]
	\arrow[from=2-4, to=1-3]
	\arrow[from=2-2, to=1-3]
	\arrow[from=3-5, to=2-4]
	\arrow[from=3-1, to=2-2]
	\arrow["\lrcorner"{anchor=center, pos=0.125, rotate=-45}, draw=none, from=1-3, to=3-3]
\end{tikzcd}\]
These are in general not isomorphic, $(B/A)^2 \not \cong B^2/A^2$. Note that, in contrast, this the analogous identity $(X/\fib D) \times (Y/\fib E) \cong (X \times Y)/(\fib D \times \fib E)$ \emph{is} valid for quotients of vector spaces (see Section~\ref{sec:linrel}).
\end{proof}

\subsection{Interconnection}\label{app:interconnection}

We will make frequent use of the projections 
\begin{itemize}
	\item $X: \R^\ell \times \R^m \times \R^n \to \R^\ell$, $X(x, y, z) = x$.
	\item $Y: \R^\ell \times \R^m \times \R^n \to \R^m$, $Y(x, y, z) = y$.
	\item $Z: \R^\ell \times \R^m \times \R^n \to \R^n$, $Z(x, y, z) = z$.
\end{itemize}
and use expressions like $\{X \in U\}$ to mean 
$$
\{(x, y, z) \in \R^\ell \times \R^m \times \R^n \mid X(x, y, z) \in U\}.
$$

Let $e: \R^\ell \to \R^m$ be a morphism in the category $\gaussex$ represented by the decorated copartial function
$$
\left(
\begin{tikzcd}
	& \R^i &                    \\
	\R^\ell \arrow[ru, "A"] &   & \R^m \arrow[lu, "B"']
\end{tikzcd}
, \:
\epsilon: \R^0 \to \R^i
\right).
$$
Then $e$ corresponds to the probability space $\Sigma = \left(\R^\ell \times \R^m, \mathcal{E}, P \right)$, where
$$
\mathcal{E} = \left\{ \{(x, y) \mid By - Ax \in U \} \mid U \subset \R^i \: \mathrm{Borel}\right\}
$$
and
$$
P\left( \{(x, y) \mid By - Ax \in U \} \right) = \epsilon(U)
$$
for all $U \subset \R^i$ Borel. If $f: \R^m \to \R^n$ is another morphism in $\gaussex$, and if $\Phi = \left( \R^m \times \R^n, \mathcal{F}, Q \right)$ is the stochastic system corresponding to $f$, then the systems $\Sigma$ and $\Phi$ satisfy the following property.

\begin{thm}
	Let us define stochastic systems $\Sigma^\prime = \left( \R^\ell \times \R^m \times \R^n, \mathcal{E}^\prime, P^\prime \right)$ and $\Phi^\prime = \left( \R^\ell \times \R^m \times \R^n, \mathcal{F}^\prime, Q^\prime \right)$ in the following way. 
	\begin{itemize}
		\item $\mathcal{E}^\prime = \{E \times \R^n \mid E \in \mathcal{E}\}$.
		\item $\mathcal{F}^\prime = \{\R^\ell \times F \mid F \in \mathcal{F}\}$.
		\item For all $E \in \mathcal{E}$,  $P^\prime(E \times \R^n) = P(E)$.
		\item For all $F \in \mathcal{F}$, $Q^\prime(\R^\ell \times F) = Q(F)$.
	\end{itemize}
	Then $\Sigma^\prime$ and $\Phi^\prime$ are complementary.
\end{thm}
\begin{proof}
	Let $f$ be represented by the decorated copartial function
	$$
	\left(
	\begin{tikzcd}
		& \R^j &                    \\
		\R^m \arrow[ru, "C"] &   & \R^n \arrow[lu, "D"']
	\end{tikzcd}
	, \:
	\phi: \R^0 \to \R^j
	\right).
	$$
	Then $\Sigma^\prime$ and $\Phi^\prime$ are linear stochastic systems with fibers $\mathbb{L} = \{ AX = BY \}
	$ and $\mathbb{M} = \{ CY = DZ \}$. If we can demonstrate that $\mathbb{L} + \mathbb{M} = \R^l \times \R^m \times \R^n$, then the the stochastic systems $\Sigma^\prime$ and $\Phi^\prime$ will be complementary by Proposition~\ref{prop:linear_complementary}.

	One the one hand, since $\mathbb{L} \subset \R^l \times \R^m \times \R^n$ and $\mathbb{M} \subset \R^l \times \R^m \times \R^n$, we clearly have $\mathbb{L} + \mathbb{M} \subset \R^l \times \R^m \times \R^n$. To demonstrate the converse inclusion, we let $(x, y, z) \in \R^l \times \R^m \times \R^n$. Since the matrix $D$ is surjective, there exists a vector $x^\prime \in \R^n$ such that $Dy^\prime = Cy$. The vectors $(0, 0, z - z^\prime) \in \R^l \times \R^m \times \R^n$ and $(x, y, z^\prime) \in \R^l \times \R^m \times \R^n$ then satisfy the following properties.
	\begin{itemize}
		\item $(0, 0, z - z^\prime) \in \mathbb{L}$
		\item $(x, y, z^\prime) \in \mathbb{M}$
		\item $(0, 0, z - z^\prime) + (x, y, z^\prime) = (x, y, z)$.
	\end{itemize}
	Therefore $(x, y, z) \in \mathbb{L} + \mathbb{M}$. Since the vector $(x, y, z)$ was arbitrarily chosen, $\R^\ell \times \R^m \times \R^n \subset \mathbb{L} + \mathbb{M}$, as desired.
\end{proof}

Since $\Sigma^\prime$ and $\Phi^\prime$ are complementary, there exists a unique stochastic system $\Gamma = \left( \R^\ell \times \R^m \times \R^n, \mathcal{G}, R \right)$, called the \emph{interconnection} of $\Sigma^\prime$ and $\Phi^\prime$, with the following properties.
\begin{itemize}
	\item $\mathcal{G}$ is generated by the collection $\mathcal{C} = \{E^\prime \cap F^\prime \mid E^\prime \in \mathcal{E}^\prime, F^\prime \in \mathcal{F}^\prime \}$.
	\item For all $E^\prime \in \mathcal{E}^\prime$ and $F^\prime \in \mathcal{F}^\prime$, $R(E^\prime \cap F^\prime) = P^\prime(E^\prime)Q^\prime(F^\prime)$.
\end{itemize}
Furthermore, the following holds.

\begin{thm}
	The stochastic system $\Gamma$ corresponds to the composite morphism $g: \R^\ell \to \R^m \times \R^n$, where
	$$
	g = e \, ; \cpy_{\R^m} \, ; \left( \id \otimes f \right).
	$$
\end{thm}
\begin{proof}
Let $\Gamma^\prime = (\R^l \times \R^m \times \R^n, \mathcal{G}^\prime, R^\prime)$ be the stochastic system corresponding to the composite morphism $g$. We want to prove that $\Gamma = \Gamma^\prime$.

To begin with, recall that the stochastic system $\Sigma^\prime$ is characterized by the equations
$$
\begin{aligned}
	\mathcal{E}^\prime &= \{E \times \R^n \mid E \in \mathcal{E}\} \\
	&= \left\{ \{BY - AX \in U\} \mid U \subset \R^i \: \mathrm{Borel} \right\} 
\end{aligned}
$$
and
$$
\begin{aligned}
	P^\prime \left( \{BY - AX \in U\} \right) &= P\left( \{ (x, y) \mid By - Ax \in U \} \right) \\
	&= \epsilon(U)
\end{aligned}
$$
for all $U \subset \mathbb{R}^i$ Borel. If we let $f$ be represented by the decorated copartial function
$$
\left(
\begin{tikzcd}
	& \R^j &                    \\
	\R^m \arrow[ru, "C"] &   & \R^n \arrow[lu, "D"']
\end{tikzcd}
, \:
\phi: \R^0 \to \R^j
\right),
$$
then the stochastic system $\Phi^\prime$ may be characterized in a similar way:
$$
\mathcal{F}^\prime = \left\{ \{DZ - CY \in V\} \mid V \subset \R^j \: \mathrm{Borel} \right\} 
$$
and
$$
Q^\prime \left( \{DZ - CY \in V\} \right) = \phi(V)
$$
for all $V \subset \R^j$ Borel. Finally, recall that the $\sigma$-algebra $\mathcal{G}$ is generated by the collection $\mathcal{C}$, where
$$
\begin{aligned}
	\mathcal{C} &= \{E^\prime \cap F^\prime \mid E^\prime \in \mathcal{E}^\prime, F^\prime \in \mathcal{F}^\prime\} \\
	&= \left\{ \{BY - AX \in U \} \cap \{ DZ - CY) \in V \} \mid U \subset \R^i, V \subset \R^j \: \mathrm{Borel} \right\} \\
	&= \left\{ \left\{ (BY - AX, DZ - CY) \in U \times V \right\} \mid U \subset \R^i, V \subset \R^j \: \mathrm{Borel} \right\},
\end{aligned}
$$
and the probability measure $R$ is uniquely determined by the equation 
$$
\begin{aligned}
	R\left( \left\{ (BY - AX, DZ - CY) \in U \times V \right\} \right) &= P^\prime\left(\{BY - AX \in U \}\right)Q^\prime\left( \{ DZ - CY) \in V \} \right) \\
	&= \epsilon(U) \phi(V) \\
	&= (\epsilon \otimes \phi)(U \times V)
\end{aligned}
$$
for all $U \subset \R^i$ and $V \subset \R^j$ Borel.

On the other hand, the morphism $g$ is represented by the decorated copartial function
$$
\left(
\begin{tikzcd}
	& \R^i  \times \R^j &                               \\
	\R^\ell \arrow[ru, "A^\prime"] &       & \R^m \times \R^n \arrow[lu, "B^\prime"']
\end{tikzcd}
, \:
\epsilon \otimes \phi: \R^0 \to \R^i \times \R^j
\right),
$$
where 
$$
A^\prime = \begin{pmatrix} A \\ 0 \end{pmatrix}
$$
and
$$
B^\prime = \begin{pmatrix} B & 0 \\ -C & D \end{pmatrix}.
$$
Therefore, the stochastic system $\Gamma^\prime$ is characterized by the following equations:
$$
\mathcal{G}^\prime = \left\{ \{(BY - AX, DZ - CY) \in W\} \mid W \subset \R^i \times \R^j \: \mathrm{Borel}\right\}
$$
and
$$
R^\prime \left( \{(BY - AX, DZ - CY) \in W\} \right) = (\epsilon \otimes \phi)(W)
$$
for all $W \subset \R^i \times \R^j$ Borel. Clearly, $\mathcal{C} \subset \mathcal{G}^\prime$, so $\mathcal{G} \subset \mathcal{G}^\prime$. Furthermore, the probability measures $R$ and $R^\prime$ agree on the collection $\mathcal{C}$, so they agree on $\mathcal{G}$ as well. If we can demonstrate the converse inclusion $\mathcal{G}^\prime \subset \mathcal{G}$, then we will have proved the desired equality $\Gamma = \Gamma^\prime$.

We will use a monotone class theorem argument. Let us define collections $\mathcal{P}$ and $\mathcal{D}$ by
$$
\mathcal{P} = \{U \times V \mid U \subset \R^i, V \subset \R^j \: \mathrm{Borel}\}
$$
and
$$
\mathcal{D} = \left\{W \subset \R^i \times \R^j \: \mathrm{Borel} \mid \{ (BY - AX, DZ - CY) \in W \} \in \mathcal{G} \right\}.
$$
Then the collection $\mathcal{P}$ is a p-system, and the collection $\mathcal{D}$ is a d-system. Furthermore, $\mathcal{P} \subset \mathcal{D}$, so by the monotone class theorem, $\mathcal{D}$ contains the $\sigma$-algebra generated by $\mathcal{P}$. It is well known that $\sigma$-algebra generated by $\mathcal{P}$ is the Borel $\sigma$-algebra on $\R^i \times \R^j$. Therefore, for all $W \subset \R^i \times \R^j$ Borel,
$$
\{ (BY - AX, DZ - CY) \in W \} \in \mathcal{G},
$$
so $\mathcal{G}^\prime \subset \mathcal{G}$, as desired.
\end{proof}

Let $e: \R^\ell \to \R^m$, $f: \R^m \to \R^n$, and $g: \R^\ell \to \R^m \times \R^n$ be the morphisms defined in the preceding section, and let $\Sigma$, $\Phi$, and $\Gamma$ be the corresponding stochastic systems. Let us define a stochastic system $\Xi = \left( \R^\ell \times \R^n, \mathcal{H}, S \right)$ by
$$
\mathcal{H} = \left\{ H \subset \R^l \times \R^n \mid \{ (X, Z) \in H\} \in \mathcal{G} \right\}
$$
and
$$
S(H) = R \left(\{ (X, Z) \in H\} \right)
$$
for all $H \in \mathcal{H}$. Willems calls the system $\Xi$ the result of ``eliminating a variable from $\Gamma$.'' It satisfies the following property.
\begin{thm}
	The stochastic system $\Xi$ corresponds to the composite morphism $h: \R^\ell \to \R^n$, where
	$$
	h = f \, ; g.
	$$
\end{thm}
\begin{proof}
	Let $\Xi^\prime = (\R^l \times \R^n, \mathcal{H}^\prime, S^\prime)$ be the stochastic system corresponding to the morphism $h$. We want to prove that $\Xi = \Xi^\prime$.
	
	To begin with, let $f$ be represented by the decorated copartial function
	$$
	\left(
	\begin{tikzcd}
		& \R^j &                    \\
		\R^m \arrow[ru, "C"] &   & \R^n \arrow[lu, "D"']
	\end{tikzcd}
	, \:
	\phi: \R^0 \to \R^j
	\right).
	$$
	Then $h$ is represented by a decorated copartial function
	$$
	\left(
	\begin{tikzcd}
		&& \R^k \\
		& \R^i & {} & \R^j \\
		\R^\ell && \R^m && \R^n
		\arrow["A", from=3-1, to=2-2]
		\arrow["B"', from=3-3, to=2-2]
		\arrow["M", dashed, from=2-2, to=1-3]
		\arrow["N"', dashed, from=2-4, to=1-3]
		\arrow["D"', from=3-5, to=2-4]
		\arrow["C", from=3-3, to=2-4]
		\arrow["\lrcorner"{anchor=center, pos=0.125, rotate=-45}, draw=none, from=1-3, to=2-3]
	\end{tikzcd}
	,
	\:
	M\epsilon + N\phi: \R^0 \to \R^k
	\right),
	$$
	where
	$$
	\mathrm{ker} \begin{pmatrix} M & N \end{pmatrix} = \mathrm{im} \begin{pmatrix} B \\ -C \end{pmatrix}.
	$$
	Therefore, the stochastic system $\Xi^\prime$ is characterized by the equations
	$$
	\mathcal{H}^\prime = \left\{ \{(x, z) \mid NDz - MAx \in T\} \mid T \subset \R^k \: \mathrm{Borel} \right\}
	$$
	and
	$$
	S \left( \{(x, z) \mid NDz - MAx \in T\} \right) = (M\epsilon + N\phi)(T)
	$$
	for all $T \subset \R^k$ Borel.
	
	We will first demonstrate that $\mathcal{H}^\prime \subset \mathcal{H}$, and that $S$ and $S^\prime$ agree on $\mathcal{H}^\prime$. To that end, let $H^\prime \in \mathcal{H}^\prime$. If we can prove that
	$\{(X, Z) \in H^\prime\} \in \mathcal{G}$ and $S^\prime(H^\prime) = S(H^\prime)$, then the results will follow from the arbitrary selection of $H^\prime$. Per the characterization above, there exists a Borel subset $T \subset \R^k$ such that $H^\prime = \{(x, z) \mid NDz - MAx \in T\}$. Hence,
	$$
	\begin{aligned}
		\{(X, Z) \in H^\prime\} &= \{NDZ - MAX \in T\} \\
		&= \{M(BY - AX) + N(DZ - CY) \in T\} \\
		&= \left\{(BY - AZ, DZ - BY) \in \{(u, v) \mid Mu + Nv \in T\}\right\} \\
		&\in \mathcal{G},
	\end{aligned}
	$$
	where the second line follows from the fact that $MB = NC$, and the third line follows from the fact that the set $\{(u, v) \mid Mu + Nv \in T\}$ is Borel. Furthermore,
	$$
	\begin{aligned}
		S^\prime(H^\prime) &= (M\epsilon + N\phi)(T) \\
		&= R\left(\{M(BY - AX) + N(DZ - CY) \in T\} \right) \\
		&= R\left \{(X, Z) \in H^\prime\} \right) \\
		&= S(H^\prime),
	\end{aligned}
	$$
	where the second line follows from the fact that the random variables
	$$
	BY - AX: \left(\R^\ell \times \R^m \times \R^n, \mathcal{G}\right) \to \left(\R^i, \mathcal{B}(\R^i)\right)
	$$
	and
	$$
	DZ - CY: \left(\R^\ell \times \R^m \times \R^n, \mathcal{G}\right) \to \left(\R^j, \mathcal{B}(\R^j)\right)
	$$
	are independent with respect to the probability measure $R$.
	
	If we can demonstrate the converse inclusion $\mathcal{H} \subset \mathcal{H}^\prime$, then we will have proved the desired equality $\Xi = \Xi^\prime$. To that end, let $H \in \mathcal{H}$. Then, by definition of $\mathcal{H}$, $\{(X, Z) \in H\} \in \mathcal{G}$, so there exists a Borel subset $W \subset \R^i \times \R^j$ such that
	$$
	\{(X, Z) \in H\} = \{(BY - AX, DZ - CY) \in W\}.
	$$
	We will prove that there exists a Borel subset $T \subset \R^k$ such that
	$$
	W = \{(u, v) \mid Mu + Nv \in T\}.
	$$
	To see this let $(u, v) \in W$. Since the matrices $B$ and $D$ are surjective, so are functions $BY - AX$ and $DZ - CY$, and there exists a vector $(x, y, z) \in \{(X, Z) \in H\}$ such that $By - Ax = u$ and $Dz - Cy = v$. Furthermore, for all $y^\prime \in \R^m$, we also have
	$$
	(x, y + y^\prime, z) \in \{(X, Z) \in H\},
	$$
	so
	$$
	(u + By^\prime, v - Cy^\prime) \in W.
	$$
	Since the vectors $u$, $v$, and $y^\prime$ were arbitrarily chosen, the set $W$ must satisfy.
	$$
	W = W + \mathrm{im} \begin{pmatrix} B \\ -C \end{pmatrix}.
	$$
	The existence of $T$ then follows from the fact that
	$$
	\mathrm{ker} \begin{pmatrix} M & N \end{pmatrix} = \mathrm{im} \begin{pmatrix} B \\ -C \end{pmatrix}.
	$$
	Therefore,
	$$
	\begin{aligned}
		\{(X, Z) \in H\} &= \{(BY - AX, DZ - CY) \in W\} \\
		&= \left\{(BY - AX, DX - CY) \in \{(u, v) \mid Mu + Nv \in T\} \right\} \\
		&= \left\{(X, Z) \in \{(x, z) \mid NDz - MAx \in T\}\right\}.
	\end{aligned}
	$$
	This implies that $H = \{(x, z) \mid NDz - MAx \in T\} \in \mathcal{H}^\prime$. The inclusion $\mathcal{H} \subset \mathcal{H}^\prime$ follows from the arbitrary selection of $H$.
\end{proof}

\end{document}